\documentclass[lettersize,journal]{IEEEtran}
\usepackage{amsmath,amsfonts}
\usepackage{algorithmic}
\usepackage{algorithm}
\usepackage{array}
\usepackage[caption=false,font=normalsize,labelfont=sf,textfont=sf]{subfig}
\usepackage{textcomp}
\usepackage{stfloats}
\usepackage{url}
\usepackage{verbatim}
\usepackage{graphicx}
\usepackage{cite}
\usepackage[utf8]{inputenc} 
\usepackage{booktabs} 
\usepackage{colortbl, amsmath}
\usepackage{graphics} 
\usepackage{epsfig} 
\usepackage{mathptmx} 
\usepackage{times} 
\usepackage{amssymb}  
\usepackage{moreverb}
\usepackage{subfig}
\usepackage{caption}
\usepackage{amsmath,amsfonts}
\usepackage{mathtools}

\usepackage{amsmath,amsfonts,amsthm}
\usepackage{algorithmic}
\usepackage{algorithm}
\usepackage{array}
\usepackage[caption=false,font=normalsize,labelfont=sf,textfont=sf]{subfig}
\usepackage{textcomp}
\usepackage{stfloats}
\usepackage{url}
\usepackage{verbatim}
\usepackage{graphicx}
\usepackage{cite}
\newtheorem{thm}{Theorem}
\newcommand{\SE}{\text{SE}(3)}
\newcommand{\SO}{\text{SO}(3)}
\newcommand{\R}{\mathbb{R}}
\newcommand{\I}{\mathbb{I}}

\newcommand{\se}{\mathfrak{se}(3)}

\newcommand{\so}{\mathfrak{so}(3)}

\newcommand{\bB}{\mathbb{B}}
\newcommand{\Ad}{\mathbf{Ad}}
\newcommand{\ad}{\mathbf{ad}}
\newcommand{\0}{\mathbf{0}}
\newcommand{\1}{\mathbf{1}}
\newcommand{\tr}{\textbf{tr}}
\newcommand{\sk}{\textbf{sk}}

\newcommand{\diag}{\text{diag}}

\newcommand{\bI}{\mathbb{I}}

\newcommand{\der}{\mathbf{d}}
\newcommand{\Der}{\mathbf{D}}
\newcommand{\Hes}{\mathbf{H}}

\newcommand{\bE}{\mathbb{E}}
\newcommand{\bg}{\mathbb{g}}

\newcommand{\fg}{\mathfrak{g}}

\newcommand{\cB}{\mathcal{B}}
\newcommand{\fX}{\mathfrak{X}}
\newcommand{\cC}{\mathcal{C}}
\newcommand{\cL}{\mathcal{L}}
\newcommand{\cE}{\mathcal{E}}

\newtheorem{lemma}{Lemma}
\newtheorem{proposition}{Proposition}
\newtheorem{corollary}{Corollary}

\theoremstyle{definition}
\newtheorem*{example*}{Example}
\newtheorem{assumption}{Assumption}
\newtheorem{remark}{Remark}

\newtheorem{problem}{Problem}
\newtheorem{definition}{Definition}
\begin{document}

\title{Geometric Fault-Tolerant Neural Network Tracking Control of Unknown Systems on Matrix Lie Groups}

\author{Robin Chhabra$^{1}$~\IEEEmembership{Senior Member,~IEEE,} and Farzaneh Abdollahi$^{2}$~\IEEEmembership{Senoir Member,~IEEE,}
\thanks{*This work was supported by a grant from the Natural Sciences and Engineering Research Council of Canada (DGECR-2019-00085).}
\thanks{$^{1}$R. Chhabra (corresponding author) is with the Department of Mechanical, Industrial, and Mechatronics Engineering, Toronto Metropolitan University,
        Toronto, ON M5B 2K3, Canada
        {\tt\small robin.chhabra@torontomu.ca}}%
\thanks{$^{2}$F. Abdollahi is with the
Department of Electrical Engineering,
        Amirkabir University of Technology,
Tehran, Iran
        {\tt\small 	f\_abdollahi@aut.ac.ir}}}%

\markboth{IEEE Transactions on Automatic Control}%
{Chhabra \MakeLowercase{\textit{et al.}}: Neural Network Control of
Unknown Systems on Matrix Lie Groups}


\maketitle

\begin{abstract}
We present a geometric neural network-based tracking controller for systems evolving on matrix Lie groups under unknown dynamics, actuator faults, and bounded disturbances. Leveraging the left-invariance of the tangent bundle of matrix Lie groups, viewed as an embedded submanifold of the vector space $\R^{N\times N}$, we propose a set of learning rules for neural network weights that are intrinsically compatible with the Lie group structure and do not require explicit parameterization. Exploiting the geometric properties of Lie groups, this approach circumvents parameterization singularities and enables a global search for optimal weights. The ultimate boundedness of all error signals---including the neural network weights, the coordinate-free configuration error function, and the tracking velocity error---is established using Lyapunov’s direct method. To validate the effectiveness of the proposed method, we provide illustrative simulation results for decentralized formation control of multi-agent systems on the Special Euclidean group.\end{abstract}

\begin{IEEEkeywords}
Neural Network Control, Matrix Lie Group, Formation Control, Decentralized Control
\end{IEEEkeywords}

\section{Introduction}
\IEEEPARstart{C}{ontrolling} nonlinear systems with no accurate knowledge of their dynamics is crucial, since in most practices accurate system identification is either expensive or infeasible. Adaptive control methods tackle this problem by offering stabilizing solutions for systems with models containing time-varying and uncertain parameters. Among all, data-driven techniques are lately gaining increasing attention due to the advancements in the speed and computational power of processors \cite{raissi-18}.
 These techniques, and more specifically Neural Networks (NNs), are known as powerful tools for function approximation, applicable to identifying and generalizing system
dynamics, motion planning, and control \cite{Kasra-book,wiliams:17,our-book,janathan:13,Duong:21}.
The challenge, however, is to stably train NN weights to provide convergence to (globally) optimum solutions even though the system undergoes drastic motion and variation in dynamic properties. Hence, developing model-free NN learning rules defined on the whole system's configuration space is desirable.  

The nonlinear dynamics of most examples of multi-agent systems including multi-vehicles evolve on matrix Lie groups, often associated with the space of rigid transformations, i.e., the Special Euclidean group ($\SE$) or its Lie subgroups. An intrinsic approach for modeling dynamical systems on matrix Lie groups uses the Euler-Poincar\'{e} equation~\cite{bloch,bullo-book}, leading to a set of matrix differential equations. Further, an effective functional analysis respecting the Lie group structure of the configuration manifold avoids local coordinate assignment and the resulting possibly ill-conditioned Jacobians~\cite{Lee-book17}. 
In fact, adaptive NN control of multi-agent systems based on such coordinate-free modeling 
 and analysis techniques define globally valid learning rules that can converge to global optima.

This paper develops a stable geometric NN control strategy for systems represented on matrix Lie
groups. We assume the system model is unavailable but obeys the Euler-Poincar\'{e} equation. The control problem we solve is tracking an arbitrary trajectory in the presence of matched and unmatched uncertainties, with a known desired performance. We consider unmatched uncertainties in the form of loss of actuator efficiency to model fault.
The contributions of this paper are:
  \begin{enumerate}
      \item We provide a geometric framework for coordinate-free functional analysis essential for function approximation and back-propagation on matrix Lie groups.
      \item We develop a set of global model-free learning rules for fault-tolerant NN-based tracking control and provide the necessary conditions for its stability, in the presence of bounded disturbances and uncertainties, using Lyapunov's direct method.
      \item We extend our approach to decentralized NN-based formation control of multi-agent systems evolving on matrix Lie groups.
  \end{enumerate}
To derive the intrinsic representation of the learning rules, we develop a geometric computational framework for back-propagation on matrix Lie groups, utilizing matrix differential equations rather than vector-based counterparts. Eliminating the parameterization of the configuration manifold further reduces computational complexity by avoiding the need to compute intricate Jacobians associated with parameterized representations in the learning process. By preserving the Lie group structure of the configuration manifold, the proposed learning rules also enable the global optimization of NN weights and facilitate the almost global stability of the controller under mild assumptions. 
We exploit the left invariance of the tangent bundle to formulate both the system dynamics and learning rules on the Lie algebra. Specifically, for the mobile robots evolving on $\SE$, this formulation ensures that velocities remain independent of the configuration while remaining compatible with onboard sensor measurements. This inherent geometric structure enhances fault tolerance and robustness in NN-based learning and control. Consequently, our approach establishes an NN control framework with improved generalizability, stability, and convergence properties.
  
  This paper is organized as follows. Section \ref{sec:relwork} provides an overview of the existing literature on the topic. In Section \ref{prelim}, we review some preliminaries and state the problem. The main results are provided in Section \ref{NNcont}. Section \ref{NNcont_mtl} extends our results to the decentralized formation control of multi-agent systems and demonstrates the efficacy of our approach in simulation. Finally, Section \ref{conclu} concludes the paper.

 \section{Related Work}\label{sec:relwork}
\subsection{NN-based Adaptive Control}

Stable NN-based adaptive control of uncertain nonlinear systems has been a topic of interest since the 1980s. In a seminal work, Chen and Khalil introduced layered neural networks to control feedback linearizable systems \cite{chen1992adaptive}. A variable neural network architecture is proposed for approximating the unknown nonlinearities of dynamical systems where the number of basis functions can be either increased or decreased in time to accommodate the problem of overfitting or underfitting the data set \cite{liu1999variable}. 

A NN-based control for a class of nonaffine nonlinear systems with input saturation was introduced in \cite{esfandiari2014adaptive}. A real-time deep NN adaptive control architecture is also developed for uncertain control-affine nonlinear systems to track a time-varying desired trajectory \cite{le2021real}. Authors of \cite{Kasra-book} investigate stable NN-based control of different forms of affine, nonaffine, and noncanonical control systems.
\cite{Yin22StableNN5}, introduced an approach to analyze  the stability of a class of nonlinear system  systems with neural network controllers.
Considering deep residual neural networks as a nonlinear control, it has been shown  the power of universal approximation  of these networks in \cite{Tabuada23deepResnet}.  Fabiani \textit{et al.} in \cite{Fabiani23stableNN} proposed an approach to  design minimum complexity NNs with rectified linear units (ReLUs) that keeps the desirable properties of a given Model Predictive Control scheme. 

Some works have been reported that consider the structure of the system's equations of motion in the design of NN-based controllers. In \cite{Lutter:19}, an NN-based control is developed for
systems evolving under the Euler-Lagrange equations. 
 In this approach, the kinetic and potential energy terms are modeled by neural networks, separately.
In \cite{Greydanus:19}, the concept of Hamiltonian NNs is proposed for Hamiltonian control systems. To approximate the  Hamiltonian, such NNs are designed whose weights are updated by minimizing the discrepancy between its synthetic gradients and the time derivatives of the states.  Hamiltonian and Lagrangian equations of motion, however,   rely on the generalized coordinates parameterizing the configuration manifold of the system, e.g., Euler angles to parameterize rotations, which cannot globally capture the system dynamics.
In \cite{Ruthotto-19}, a Partial Differential Equation (PDE) interpretation
of Convolutional Neural Networks (CNNs) is proposed and two new classes of hyperbolic and parabolic CNNs for residual NNs
and  ResNet architectures guided by the governing PDEs are derived.

\subsection{Control Systems on Lie Groups}
States of many systems, including mobile robots and manipulators, can be unambiguously represented on matrix Lie groups. More specifically in robotics, these Lie groups are often in the form of a Cartesian product of Lie subgroups of $\SE$ whose elements represent relative rigid body poses (position and rotation). 
Several fundamental studies on the role of Lie groups and symmetry in the kinematic and dynamic analysis of rigid robots were conducted in the literature \cite{murray,lynch2017modern,park1995lie,stramigioli2002geometry,brockett1984robotic} that have been followed by more contemporary researchers \cite{Chhabra2014a,muller2018screw}. Accordingly, by taking advantage of various properties of Lie groups some efficient control strategies for rigid multi-body systems have been developed \cite{bloch,borna,chhabracontrol,Ghasemi:2019,ABD_CHH:22,bullo-book,BULLO199917}. Specifically, feedback linearization on the Lie algebra of $\SE$ was proposed for vehicle-manipulator systems to perform full-pose end-effector precision control and target tracking \cite{chhabracontrol,borna,patrik1,patrik2}.
The Lie group $\SE$ is the semi-direct product of the Special Orthogonal group $\SO$ and $\R^3$.
In most studies, parametrizations of $\SO$, e.g., Rodriguez parameters, or Euler angles, are applied to formulate the dynamics on $\R^6$, which usually face limitations. Examples include the gimbal lock singularities in Euler angles \cite{markley-book} or unwinding phenomenon in quaternion-based feedback control \cite{mayhew:12}. In the case of using inverse dynamics in control design, the system becomes unstable at these singularities \cite{moghaddam2022}.  

Inspired by rigid body motion in $\SE$, the theory of control systems on Lie groups was pioneered in the work of Brockett \cite{doi:10.1137/0310021}, followed by Jurdjevic and Sussmann \cite{jurdjevic1972control}, where various control theoretic notions were extended to right-invariant vector fields. These early studies laid the groundwork for investigating various control structures in the context of differentiable manifolds and Lie groups. Motion control of drift-free, left-invariant systems on Lie groups was proposed in \cite{leonard1995motion}. Bullo and Lewis studied configuration error functions, connections, and extension of PD controllers to Lie groups \cite{bullo-book}.  
Zhang \textit{et al.} have extended the concept of integral control to systems on Lie groups \cite{zhang2015integral}, and Berg \textit{et al.} have investigated the stability of output-tracking PID controllers with intrinsic integral action\cite{maithripala2015intrinsic}. Further, Bloch \textit{et al.} have published research on the optimal control on Lie groups \cite{hussein2005optimal,bloch2017optimal}.
In \cite{Duong:21}, a neural ordinary differential equation network is designed for the identification of Hamiltonian systems on $\SE$, whose model was presented in \cite{Lee-book17}. However, no mathematical proof was provided to support tracking convergence of the closed-loop system when the neural network identifier is in the loop.
 In \cite{hashim-22}, an NN stochastic filter-based controller was proposed for systems evolving on $\SO$, which is not extendable to general matrix Lie groups.

\subsection{Adaptive Control of Multi-Agent Systems}
Multi-agent systems offer a robust and economical solution to cooperatively accomplish missions in a wide range of applications from smart cities, smart grids, and wildlife tracking to underwater and space exploration.
Using neighbors' information has led to the development of several decentralized controllers for different missions,  such as consensus,  swarming, formation,  synchronization,  and coverage  \cite{wang-peng:16,rezaee:18-2,parapari-2,forugh11,Mei-Ren:15}. There is a large body of research on data-driven approaches, specifically NNs in recent years, for adaptive control of uncertain multi-agent systems \cite{kiumarsi2017optimal}. Adaptive consensus control of nonlinear stochastic multi-agent systems was investigated, where the NN was adopted to estimate the unknown nonlinear dynamics \cite{li2021adaptive}. Multiple fixed-time adaptive NN controllers were proposed for nonlinear multi-agent systems \cite{huang2022adaptive} that can be stochastic and non-affine \cite{bai2022neural}.

 Distributed NNs were developed for adaptive synchronization \cite{peng2013distributed,shen2019neural} and consensus control \cite{yue2020neural} of uncertain multi-agent systems. These results have been enhanced to include actuator and network faults \cite{jin2021adaptive,qin2019neural}, as well as guaranteed transient performance \cite{cao2020performance}. 
A NN-based consensus control was introduced for a class of nonlinear multi-agent systems in \cite{Zou24masnn}.
Multi-vehicle systems are often described as several rigid bodies (vehicles) whose relative motion relative to one another can be captured on Lie groups.
  An attitude synchronization of multiple rigid bodies is introduced in  \cite{shahoo:14}, and its convergence is analyzed locally around a consensus configuration. In \cite{fadakar:14}, a robust adaptive attitude synchronization law for a network of rigid bodies is developed.  However, the studied convergence to the consensus configuration is
 merely limited to the specific tree topologies of the network graph. A Sliding Mode Control (SMC)
using the exponential coordinates was designed for spacecraft formation control problems with a virtual leader in  \cite{Lee:13}.
In \cite{Ghasemi:2019}, a robust formation control is proposed for multiple rigid bodies in the presence of uncertainties and disturbances. An adaptive super-twisting SMC with an intrinsic PID sliding surface was developed to cope with the uncertainties, unmodelled dynamics, and external disturbances while keeping the desired formation.
Most introduced control approaches in this section and the preceding sections either consider the full or partial knowledge of the system model or
 require configuration parameterization for the numerical implementation. Therefore, there is a gap in the literature to offer effective control of unknown multi-agent robotic systems described on matrix Lie groups.

  \section{Preliminaries and Problem Statement} \label{prelim}
%

\subsection{Notation}
Let $G$ be a matrix Lie group and $\fg$ be its Lie algebra that is equipped with the Lie bracket $[\cdot,\cdot]\colon\fg\times\fg\rightarrow\fg$ defined based on the matrix commutator operator. 
Let $\exp\colon\fg\rightarrow G$ be the exponential map associated with the Lie group $G$. We use the same notation for the exponential map of any Lie group and for brevity, we sometimes write $e^{\xi}\coloneqq\exp(\xi)$ for every $\xi\in\fg$. 
The tangent and cotangent bundles of $G$ are coined by $TG$ and $T^*G$, respectively. Generally, $^*$ specifies the dual of a vector space, a vector (field), or a vector bundle over $G$.
An $\R^k$-valued function $f\colon G\to\R^k$ belongs to the class $\cC^p(G)$, if it has $p^{th}$ derivative which is continuous. In this paper, wherever the codomain dimension of $f$ is not specified, it is assumed to be one. 
We specify the space of all vector fields on $G$ by $\fX(G)$.The directional derivative of a function $f\in\cC^{\infty}(G)$ along any $X\in\fX(G)$ is defined by the Lie derivative operator $\cL_Xf$. We also denote the total derivative of a function by $\der$ and the Jacobian (gradient) of a function between two Euclidean spaces by $\nabla$.
In addition, we specify the natural pairing between the elements of $T^*G$ and $TG$ by $\left<\cdot|\cdot\right>$ and the pairing between elements of $\fg$ and $\fg^*$ by $\left<\cdot,\cdot\right>$.

In this paper, we use the Frobenius norm of matrices, i.e., for any matrix $B=[B_{ij}] \in \R^{M \times K}$, we define
\begin{equation}
  \|B\|^2=\tr(B^TB)=\sum_{i,j} B_{ij}^2,
\end{equation}
where the operator $\tr$ is the trace of a square matrix. Further, by $\sk$ we denote the operator that generates the skew-symmetric component of a square matrix, i.e., $\sk(B)=\frac{1}{2}(B-B^T)$ for all $B\in\R^{M\times M}$.
 We also define the operator $\breve{(\cdot)}\colon \R^{M\times K}\rightarrow \R^{MK}$ such that  
\begin{align}\label{eq:breve}
    \breve{B}=[B_{11}~~\cdots~~B_{M1}~~\cdots ~~ B_{1K} ~~ \cdots ~~ B_{MK}]^T\in\R^{MK}
\end{align}
that reshapes the matrix to a vector. Given the matrix dimensions, the inverse operator
$\breve{\breve{(\cdot)}}\colon \R^{MK}\rightarrow \R^{M\times K}$ reshapes an $MK$-dimensional vector to an $M\times K$ matrix. Accordingly, for any arbitrary function $F$ dependant on the matrix $B\in\R^{M\times K}$, the gradient is calculated by
\begin{align}\label{eq:defgradF}
   \nabla_{B} F \coloneqq \left(\nabla_{ \breve{B}}F\right)\breve\breve\in\R^{M\times K}.
\end{align} 
 Lastly, $\1$ and $\0$ everywhere denote the identity and zero matrices with appropriate dimensions, respectively.
\subsection{Matrix Lie Groups}
Let $G$ be an  $n$-dimensional matrix Lie group, i.e., for some $N$ it is a closed Lie subgroup of the General Linear group $\mathbb{GL}(N)\subset \R^{N\times N}$, considering the subspace topology. We also assume that $G$ is closed in $\R^{N\times N}$, the $N^2$-dimensional space of matrices that is isomorphic to $\R^{N^2}$ as a vector space, equipped with the standard topology\footnote{Let $g_l$ be any sequence of matrices in $G$, such that $g_l\rightarrow g$ as $l\rightarrow\infty$. Under the assumptions stated here, $g$ must be in $G$.}. Therefore, $G$ is an embedded submanifold of both $\mathbb{GL}(N)$ and $\R^{N\times N}$, which allows $G$ to borrow different geometric structures from these two ambient manifolds. 

\begin{lemma}[Extension of continuous functions]\label{lem:cont}
    Let $f\in\mathcal{C}^0(G)$ be any continuous ($\R^k$-valued) function on the Lie group $G\subset \R^{N\times N}$. There exists a continuous extension of this function $\bar f\in\mathcal{C}^0(\R^{N\times N})$ such that $\bar f (g)= f(g)$ for all $g\in G$.
\end{lemma}
\begin{proof}
    Since $G$ is closed in $\R^{N\times N}$, we can apply the Tietze Extension Theorem \cite[Th. 5.1, p. 149]{topology}.
\end{proof}

\begin{lemma}[Extension of smooth functions]\label{lem:smooth}
    Let $f\in\mathcal{C}^\infty(G)$ be any smooth ($\R^k$-valued) function on the Lie group $G\subset \R^{N\times N}$. There exists a smooth extension of this function $\bar f\in\mathcal{C}^\infty(\R^{N\times N})$ such that $\bar f (g)= f(g)$ for all $g\in G$.
\end{lemma}
\begin{proof}
    The proof is the direct consequence of Lemma \ref{lem:cont} and Corollary 6.27 in \cite[p. 141]{lee-man}.
\end{proof}

Indeed, these two lemmas only deal with the existence of the extension of functions that may not be unique. They are instrumental in proving results regarding the NN approximation of functions on matrix Lie groups and finding their gradients to develop effective learning laws without using any minimal coordinate chart assignments. 

Let $\fg\subset \R^{N\times N}$ (an $n$-dimensional vector subspace) be the Lie algebra of $G$ endowed with the Lie bracket $[\cdot,\cdot]\colon\fg\times\fg\rightarrow\fg$. The Lie algebra is geometrically the span of the tangent vectors at the identity element $\1_{N\times N}\in G$ that comes with a vector space isomorphism, denoted by $\wedge\colon\R^n\rightarrow\fg$ such that $\forall \xi\in\R^n$, $\xi^\wedge\in\fg$. We also define the inverse isomorphism by $\vee\colon\fg\rightarrow\R^n$ such that $(\xi^\wedge)^\vee=\xi\in\R^n$. We will use the same symbols to indicate these isomorphisms for every Lie algebra throughout this paper. 

Thanks to the Lie group structure of $G$, its tangent bundle $TG$ can be globally left-trivialized by $\fg$, i.e., for every $(g,v)\in TG\subset (\R^{N\times N})^2$:
\begin{equation}
    (g,v)\mapsto (g,g^{-1}v)\in G\times\fg\mapsto (g,(g^{-1}v)^\vee)\in G\times\R^n,\label{eq:triv}
\end{equation}
to have $TG\cong G\times\fg\cong G\times\R^n$. Given a basis for $\fg$, i.e., a set of linearly independent vectors $\{\eta_i^\wedge\in\fg|\,\,\eta_i\in\R^n,\,\, i=1,\ldots n\}$, $\forall g\in G$ we have a natural basis $\{Y_i\in T_gG|\,\,Y_i=g\eta_i^\wedge,\,\,\eta_i\in\R^n,\,\, i=1,\ldots n\}$ for the tangent space $T_gG$, whose elements define left-invariant vector fields on $G$. Therefore, any smooth vector field $X\in\fX(G)$ can be written as $\sum_{i=1}^{n}X_i(g)Y_i(g)$, where for every $i=1,\ldots,n$, 
\begin{align*}
    X_i\colon& G\rightarrow\R\\
    &g\mapsto \left<Y_i^*(g)|X(g)\right>=\left<\eta_i^*,(g^{-1}X(g))^\vee\right>
\end{align*}
is a smooth function, while $^*$ denotes the dual of a vector. Here, the pairing on the left is the canonical pairing between the tangent bundle $TG$ and the cotangent bundle $T^*G$ and the pairing on the right is the natural pairing between the $\R^n$ representation of $\fg$ and its dual.

\begin{remark}\label{rem:basis}
Although the isomorphisms $\wedge$ and $\vee$ are independent of the choice of basis for $\fg$, a natural basis for $\fg$ is the image of the canonical basis of $\R^n$ under $\wedge$. This is often implicitly assumed when treating matrix Lie groups.
\end{remark}

Let $\exp\colon\fg\rightarrow G$ be the exponential map of the Lie group $G$ that coincides with the matrix exponential defined for $\mathbb{GL}(N)$. Then $\forall g\in G$, the matrix form of the adjoint representation of the Lie group $\Ad_g\colon\R^n\rightarrow\R^n$ is defined for every $\xi\in\R^n$ as:
\begin{align}
    \Ad_g\xi:=\left.\frac{d}{d\epsilon}\right|_{\epsilon=0}\big(g e^{\epsilon\xi^\wedge}g^{-1}\big)^{\vee}=\big(g\xi^\wedge g^{-1}\big)^{\vee}.
\end{align}
Further $\forall\xi,\eta\in\R^n$, the matrix form of the adjoint representation of the Lie algebra $\ad_\xi\colon\R^n\rightarrow\R^n$ is defined based on the Lie bracket operator:
\begin{align}
    \ad_\xi\eta:=[\xi^\wedge,\eta^\wedge]^\vee=\big(\xi^\wedge \eta^\wedge-\eta^\wedge \xi^\wedge\big)^\vee.
\end{align}
Then, the dual maps of the two adjoint representations coincide with the transpose of these matrices.

Let $\der$ denote the exterior derivative (total differential) operator that is applied to a function, i.e., for every $\R$-valued function $f\in\cC^{\infty}(G)$ and a smooth vector field $X\in\fX(G)$ we have $\der f\subset T^*G$ such that $\left<\der f|X\right>=\cL_Xf$ is the Lie derivative. 
\begin{definition}\label{def:der}
  On the matrix Lie group $G$, we define the operator $\der^L$ to be the left translation of $\der$ to the identity element, i.e., $\forall f\in\cC^\infty(G)$ ($\R$-valued function) and $\forall (g,v)\in TG$ such that $v=g\xi^\wedge$:
\begin{align}\label{eq:dir-der}
    \left<\der^L_g f,\xi\right>\coloneqq\left<\der_gf|v\right>\in\R.
\end{align}
  \qed
\end{definition}
 
\begin{remark}
    In other words, \eqref{eq:dir-der} defines the directional derivative of $f$ in the direction of $v\in T_gG$, in the dual basis associated with the left trivialization of $TG$. 
\end{remark}
\begin{proposition}
    Let $f\in\cC^\infty(G)$ be a smooth $\R^k$-valued function and $\{\eta_i^\wedge\in\fg|\,\,\eta_i\in\R^n,\,\, i=1,\ldots n\}$ be a basis for $\fg$. For any smooth (local) extension of $f$ denoted by $\bar f\in\cC^\infty(\R^{N\times N})$ and $\forall g\in G$, the components of $\der^L_g f_j\in(\R^n)^*$ (for all $j=1,\ldots,k$) in the dual basis $\{\eta^*_i\in(\R^n)^*|\,\,\eta_i\in\R^n,\,\, i=1,\ldots n\}$ can be calculated by
    \begin{align}   \label{eq:derf}     (\der^L_gf_j)_i\!:=\!\left<\der^L_gf_j,\eta_i\right>\!=\!\left.\frac{d}{d\epsilon}\right|_{\epsilon=0}\!\!\!\!\!\!\!\!\!\! f_j(g e^{\epsilon \eta_i^\wedge)}\!=\!\left.\frac{d}{d\epsilon}\right|_{\epsilon=0}\!\!\!\!\!\!\!\!\!\! \bar f_j(g\!+\!\epsilon g\eta_i^\wedge),
    \end{align}
    such that $\forall \xi=\sum_{i=1}^n\xi_i\eta_i\in\R^n$, we have $\left<\der^L_g f_j,\xi\right>=\sum_{i=1}^n(\der^L_gf_j)_i\xi_i\in\R$. 
\end{proposition}
\begin{proof}
    The detailed proof is presented in Appendix \ref{app:pro-prop1}.
\end{proof}
According to this proposition, given a basis $\{\eta_i^\wedge\in\fg|\,\,\eta_i\in\R^n,\,\, i=1,\ldots n\}$ for $\fg$, we can calculate the Jacobian matrix $\Der^L_gf\colon\R^n\rightarrow\R^k$ at every element $g\in G$ by:
\begin{align}\label{eq:Der}
    \Der^L_gf=\begin{bmatrix}
        (\der^L_gf_1)_1 & \cdots & (\der^L_gf_1)_n\\
        \vdots & \ddots & \vdots\\
        (\der^L_gf_k)_1 & \cdots &(\der^L_gf_k)_n
    \end{bmatrix}.
\end{align}
Considering a curve $t\mapsto g(t)\in G$ with the velocity field $t\mapsto\xi=(g^{-1}\dot g)^\vee$, then we have $\dot f=\cL_{\dot g}f=(\Der^L_gf)\xi\in\R^k$. We also define the left-translated Hessian of an $\R$-valued function $f\in\cC^{\infty}(G)$ to be denoted by $\Hes^L_gf\coloneqq \Der_g^L((\der_g^Lf)^T)$, such that $\ddot f=\cL_{\dot g}\cL_{\dot g}f=\xi^T(\Hes_g^Lf)\xi+(\der_g^Lf)\dot\xi\in\R$. This definition of Hessian does not necessarily result in a symmetric operator, since the second derivative is taken after left translating the total derivative. Note that here we do not left translate the regular Hessian tensor of a function to identity. 

Now, let us consider any smooth function $\Phi\colon\R^k\rightarrow G\subset \R^{N\times N}$ to the Lie group $G$. 
Let $\nabla$ denote the gradient operator, i.e., for every function $\Phi$, $\nabla_\varphi\Phi\colon\R^k\rightarrow T_{\Phi(\varphi)}G$ is the linear Jacobian mapping at $\varphi\in\R^k$. This operator can be evaluated through defining a collection of $k$ smooth vector fields $\{Z_j\in\fX(G)|~Z_j=\frac{\partial\Phi}{\partial\varphi_j}\in\R^{N\times N}, j=1,\ldots,k\}$ that are always tangent to the image of $\Phi$, such that $\forall\nu\in\R^k$, we have $\nabla_\varphi\Phi(\nu)=\sum_{j=1}^k\nu_jZ_j(\Phi(\varphi))\in T_{\Phi(\varphi)}G\subset\R^{N\times N}$. The components of these vector fields in the basis associated with the left trivialization of $TG$ can then be computed by defining the functions 
\begin{align*}
    Z_{ji}\colon &\R^k\rightarrow\R\\
    &\varphi\mapsto\left<\eta_i^*,(\Phi(\varphi)^{-1}Z_j(\Phi(\varphi)))^\vee\right> 
\end{align*}
 for all $i=1,\ldots,n$ and $j=1,\ldots,k$.

\begin{definition}\label{def:JAcobian}
  On the matrix Lie group $G$, for any function $\Phi\colon\R^k\rightarrow G\subset \R^{N\times N}$ we define $\nabla^L_\varphi\Phi\colon \R^k\rightarrow\R^n$ to be the left translation of the gradient $\nabla_\varphi\Phi$ to the identity element, i.e., $\forall \varphi\in \R^k$ and given a basis $\{\eta_i^\wedge\in\fg|\,\,\eta_i\in\R^n,\,\, i=1,\ldots n\}$:
\begin{align}\label{eq:Jacobian}    \nabla^L_\varphi\Phi\coloneqq\begin{bmatrix}
        Z_{11}(\varphi) & \cdots & Z_{k1}(\varphi)\\
        \vdots & \ddots & \vdots\\
        Z_{1n}(\varphi) & \cdots &Z_{kn}(\varphi)
    \end{bmatrix}.
\end{align}
  \qed
\end{definition}
\noindent Based on this definition, then considering a curve $t\mapsto \varphi(t)\in \R^k$, we have $(\Phi^{-1}\dot\Phi)^\vee=(\Phi^{-1}\cL_{\dot\varphi}\Phi)^\vee=(\nabla^L_\varphi\Phi) \dot\varphi\in\R^n$.
\begin{remark}\hfill
\begin{enumerate}
    \item Note that other than smoothness we impose no restrictions on the definition of $\Phi$; hence, the vector fields $Z_j\in\fX(G)$ may not be everywhere linearly independent nor they may span the tangent space of $T_gG$ at every element $g\in G$ in the image of $\Phi$.

    \item In the natural basis for $\fg$ induced by the isomorphism $\wedge$ (Remark \ref{rem:basis}), the $j^{th}$ column of $\nabla^L_\varphi\Phi$ is simply $(\Phi(\varphi)^{-1}Z_j(\Phi(\varphi)))^\vee$.
\end{enumerate}
    
\end{remark}
\begin{remark}\hfill
\begin{enumerate}
    \item When the domain and codomain of a function are Euclidean spaces, both operators $\Der^L$ and $\nabla^L$ coincide with the regular Jacobian of the function. We choose to use $\nabla$ to show this Jacobian.
    \item All of the functions and operators defined in this section on the matrix Lie group $G$ only require the matrix representation of elements of $G$ and a basis for the Lie algebra $\fg$ and they are independent of any minimal parameterization of $G$.
    \item In this section, the choice of left trivialization (vs. right trivialization) was arbitrary. Every statement in this section is also valid for the right trivialization of $TG$, with minor adjustments.
\end{enumerate}    
\end{remark}

    

\subsection{NNs as Function Approximators on Matrix Lie Groups}

Relying on the universal approximation property of NNs, any continuous
function on a Euclidean space can be arbitrarily estimated using a
two-layer NN. In this section, we use this property to approximate smooth functions on matrix Lie groups.
 Let $G$ be an $n$-dimensional matrix Lie group that is embedded in $\R^{N\times N}\cong \R^{N^2}$ (through the \textit{breve} operator \eqref{eq:breve}) and $f\in\cC^\infty(G)$ be a continuous $\R^k$-valued function that has a continuous extension $\bar f\in\cC^\infty(\R^{N^2})$ (see Lemma \ref{lem:smooth}). A two-layer NN can approximate $f$ by
\begin{align}
f(g)=\bar f (g) = W \sigma (V \breve{g}) - \cE(\breve{g})\in\R^k.~~~\forall g\in G\subset\R^{N\times N}
\end{align}
Here, the matrix $V\in\R^{m\times N^2}$ contains the input to hidden layer (with $m$ fully-connected neurons) ideal interconnection
weights, the matrix $W\in\R^{k\times m}$ contains the output to hidden layer ideal interconnection weights, and $\cE\colon\R^{N^2}\rightarrow\R^k$ is a bounded approximation error function. Based on the Universal Approximation Theorem, all of the weights are bounded. The function $\sigma\colon\R^m\rightarrow\R^m$ is an activation function corresponding to the hidden layer. Although this can be any frequently used activation function, such as, sigmoid or logistic,
tangent hyperbolic, and
rectified linear unit, we work with sigmoid function in this paper that is defined by:
\begin{align}\label{eq:sigma}
  \sigma(y)\!=\!\begin{bmatrix}
      \sigma_1(y_1)~\cdots~ \sigma_m(y_m)
  \end{bmatrix}^T, \sigma_i (y_i)\! = \!-1+\frac{2}{1+e^{-2y_i}}. 
\end{align}

\subsection{Uncertain Euler-Poincar\'{e} Control Systems}

On an $n$-dimensional Lie group $G$, we consider a left-invariant metric that is induced by the metric $\bI\colon \fg\rightarrow\fg^*$ on the Lie algebra whose matrix representation belongs to $\R^{n\times n}$. 
Using the left trivialization described in \eqref{eq:triv}, the uncertain and disturbed Euler-Poincar\'{e} equation on $G$ can then be written as the following vector field \cite{bloch}:
\begin{align}
  \dot g &=g\xi^\wedge\in\R^{N\times N},\label{dynamicsm1}\\
  \dot\xi &= \bI^{-1}\ad^T_{\xi}\bI\xi+\bI^{-1}\bE u+\mu(g,\xi)+d(t)\in\R^n,
  \label{dynamicsm2}
\end{align}
where $\bI$ is an unknown symmetric, positive-definite, time-varying matrix, the vector of bounded unknown  forces is denoted by $\mu(g,\xi)$, $d(t)$
is the vector of time-dependant bounded disturbances, the variable diagonal matrix $\bE\in\R^{n\times n}$ with diagonal elements in $(0,1]$ represents the unknown actuators loss of efficiency, and $u\in\R^n$ is the vector of control inputs.

In this paper, we are concerned with designing a trajectory-tracking control strategy for a fully-actuated system that satisfies the dynamics in \eqref{dynamicsm1}-\eqref{dynamicsm2}. Let $t\mapsto g_d(t)\in G$ denote a smooth desired trajectory and $\xi_d\coloneqq (g_d^{-1}\dot g_d)^\vee\in\R^n$ be its velocity in the left-trivialized tangent bundle $TG$. We assume $g_d$, $\xi_d$, and $\dot\xi_d$ are everywhere bounded. Given a state of the system $(g,\xi)\in G\times\R^n$, we define the configuration error and accordingly, the compatible velocity error as \cite{bullo-book}
\begin{align}\label{eq:conferr}
\tilde g&\coloneqq g_d^{-1}g\in G\subset\R^{N\times N},\\ \label{eq:velerr}
    \tilde\xi&\coloneqq (\tilde g^{-1}\dot{\tilde{g}})^\vee=\xi-\Ad_{\tilde g^{-1}}\xi_d\in\R^n.
\end{align}
To form the error dynamics, we take the time derivative of \eqref{eq:velerr} and substitute $\dot\xi$ from 
\eqref{dynamicsm2}:
\begin{align}\nonumber
    \dot{\tilde\xi}&=\dot\xi+\ad_{\tilde\xi}\Ad_{\tilde g^{-1}}\xi_d-\Ad_{\tilde g^{-1}}\dot\xi_d\\ \nonumber
    &=\bI^{-1}\ad^T_{\xi}\bI\xi+\mu(g,\xi)+\ad_{\tilde\xi}\Ad_{\tilde g^{-1}}\xi_d-\Ad_{\tilde g^{-1}}\dot\xi_d+\bI^{-1}\bE u+d.
\end{align}
 Then, by introducing 
 \begin{align}\nonumber
     \tau(\tilde g, \tilde \xi,&t)\coloneqq\I^{-1}\ad^T_{\big(\tilde\xi+\Ad_{\tilde g^{-1}}\xi_d\big)}\I\big(\tilde\xi+\Ad_{\tilde g^{-1}}\xi_d\big)\\
     &+\mu(g_d\tilde g,\tilde\xi+\Ad_{\tilde g^{-1}}\xi_d)+\ad_{\tilde\xi}\Ad_{\tilde g^{-1}}\xi_d-\Ad_{\tilde g^{-1}}\dot\xi_d,
     \end{align}
     we can write the error dynamics in the compact form:
     \begin{align}
  \dot {\tilde g }&=\tilde g\tilde\xi^\wedge\in\R^{N\times N},\label{error-dynamicsm1}\\
  \dot{\tilde\xi} &= \tau+\bB u+d\in\R^n,
  \label{error-dynamicsm2}
\end{align}
where for brevity we defined $\bB=\bI^{-1}\bE$.
\begin{definition}[Configuration error function]\label{def:ConfErr}
    A smooth, proper function $\psi\colon G\rightarrow \R_{\geq 0}$ is a configuration error function if it satisfies:
    \begin{enumerate}
        \item $\psi(g^{-1})=\psi(g)$, for all $g\in G$, 
        \item $\psi(g)\geq 0$, for all $g\in G$, and $\psi(g)=0$ if and only if $g=\1_{N\times N}$, 
        \item $\der^L_{\1_{N\times N}}\psi=\0_{1\times n}$ and
        $\Hes_{\1_{N\times N}}^L\psi$ is symmetric positive-definite.
    \end{enumerate}
    \end{definition}
    Based on \cite[Prop. 6.30]{bullo-book}, there exists a constant $\Theta>0$ (the smallest critical value of $\psi$ greater than $0$) such that the positive-definite function $\psi$ is quadratic on all compact sets $\psi^{-1}(\leq\theta)$, for all $0<\theta<\Theta$, i.e., for some constants $b_2\geq b_1>0$
    \begin{align}\label{eq:b1b2}
        b_1\|\der^L_g\psi\|^2\leq\psi(g)\leq b_2\|\der^L_g\psi\|^2.
    \end{align}
    \begin{remark}
    The quadratic property is necessary to prove the stability of the proposed control law in this paper. 
    An example of a quadratic configuration error function on $\SE$ will be presented in Section \ref{NNcont_mtl}.
\end{remark}
\begin{definition}\label{def:GlobConfErr}
    A configuration error function $\psi$ is global, if we have $\psi(g)\rightarrow\infty$ while $\|g\|\rightarrow\infty$, and it is globally quadratic if $\Theta=\infty$.
\end{definition}
\begin{assumption}\label{as:quad}
    Herein, we assume that the matrix Lie group $G$ is equipped with a quadratic configuration error function $\psi$ with constants $\Theta,b_1,b_2>0$, as in \eqref{eq:b1b2}. 
\end{assumption}

 If $\tau$ and $\I$ are known, $\bE=\1_{n\times n}$, and $d=\0_{n\times 1}$, a PD-type controller based on feedback linearization on the Lie algebra $\fg$ would be \cite{bullo-book,chhabracontrol,patrik1,patrik2}:
\begin{equation}
  u^*=\I (-\tau-A \tilde{\xi}-k_p(\der^L_{\tilde g}\psi)^T)\in\R^n, \label{ustar1}
\end{equation}
where $A\in\R^{n\times n}$ is a symmetric positive-definite matrix, $k_p\geq 1$ is a constant, and $\der^L_{\tilde g}\psi$ is found based on \eqref{eq:derf}. 
Then, the nominal error dynamics for the known and undisturbed closed-loop system will be:
 \begin{align}
  \dot {\tilde g }^*&=\tilde g^*(\tilde\xi^*)^\wedge\in\R^{N\times N},\label{ideal-error-dynamicsm1}\\
  \dot{\tilde\xi}^* &= -A \tilde{\xi}^*-k_p(\der^L_{\tilde g^*}\psi)^T\in\R^n.
  \label{ideal-error-dynamicsm2}
\end{align}
By assigning the eigenvalues of $A$ and properly defining $k_p$ and parameters in the function $\psi$, we can control the behaviour of the closed-loop system. 
However, in this paper, we assume that $\tau$, $\I$, $\bE$, and $d$ are not known and accordingly we will use a neural network control to guarantee bounded deviations from the nominal error dynamics in the presence of unknown system model and bounded disturbances.

\begin{problem}[NN Tracking Control]\label{prob1}
Given the uncertain Euler-Poincar\'{e} control system on the matrix Lie group $G$ in \eqref{dynamicsm1}-\eqref{dynamicsm2}, a bounded desired trajectory $t\mapsto g_d(t)$, and the configuration and velocity errors on $G$, i.e., \eqref{eq:conferr}-\eqref{eq:velerr}, resulting in the error dynamics \eqref{error-dynamicsm1}-\eqref{error-dynamicsm2}, find the control signal $u$ based on a neural network, such that a stable tracking with the performance of the nominal closed-loop system \eqref{ideal-error-dynamicsm1}-\eqref{ideal-error-dynamicsm2} is guaranteed, in the presence of unknown bounded disturbances and uncertainties.
\end{problem}

\section{Neural Network Controller}\label{NNcont}
To design a controller for Problem \ref{prob1}, the control signal can be considered as a function defined by a NN that must approximate $\bE^{-1}u^*$ for the unknown system, as closely as possible.  The ideal weights, however, are not available, and a learning rule compatible with the Lie group structure of $G$ should be proposed to guarantee bounded approximation weight errors and also the boundedness of tracking error, in the closed-loop system. 

  Let us consider a NN with one hidden layer consisting of $m$ neurons for approximating the ideal controller for the unknown system, i.e.,
\begin{eqnarray}
  \bE^{-1}u^* = W \sigma(V x)-\cE(x)\in\R^n \label{opt-cont1}
\end{eqnarray}
where  $x\in\R^M$ is the $M$-dimensional NN input vector. The matrices $W\in\R^{n\times m}$ and $V\in\R^{m\times M}$ are the ideal NN weights with the bounds $\varpi,\vartheta>0$ such that $\|W\| \leq \varpi$  and $\|V\|  \leq \vartheta$, and $\cE(x)\in\R^n$ is the unknown approximation error which is also bounded based on the Universal Approximation Theorems. 
Since the ideal values of the weights are not known a priori, we must use the approximated weight matrices, denoted by $\hat W$ and $\hat V$. Therefore, the NN-based control signal is
\begin{eqnarray}
  u = \hat{W} \sigma(\hat{V} x)\in\R^n. \label{opt-con2}
\end{eqnarray}
Substituting \eqref{opt-con2} in \eqref{error-dynamicsm2}, the control system becomes
\begin{align}
    \dot{\tilde{\xi}} = \tau+\bB\hat{W} \sigma(\hat{V} x)+d.  
\end{align}
Based on (\ref{ustar1}) and (\ref{opt-cont1}), we can find the unknown term $\tau$ as a function of the ideal weights of the NN and approximation error that results in the following feedback transformation of the system:
\begin{align}
  \dot{\tilde{\xi}} =  -A \tilde{\xi}-k_p(\der^L_{\tilde g}\psi)^T+\bB(-W \sigma(V x)+\cE+ \hat{W} \sigma(\hat{V} x))+d.\label{er-2}
\end{align}
The error dynamics \eqref{er-2} and the error kinematics \eqref{error-dynamicsm1} define a vector field on the left-trivialized tangent bundle of $G$ that is parameterized by time $t$ (through the involvement of the desired trajectory and disturbances) and the weights $\hat W$ and $\hat V$ (by definition of the controller). Accordingly, the integral curves of this vector field are also functions of the weights $\hat W$ and $\hat V$, i.e., the time-evolution of the system can be written as $(\tilde g(t;\hat W,\hat V),\tilde\xi(t;\hat W,\hat V))$. This can be considered as the variation of the system trajectories by the NN weights. Therefore, we can make sense of the gradients of the integral curves at any time $t$ with respect to these weights that particularly define a set of vector fields along the curve $(\tilde g(t; W, V),\tilde\xi(t; W, V))$ corresponding to the ideal controller \eqref{opt-cont1}. This is indeed instrumental in defining appropriate learning rules for the NN. Before presenting such a rule, we state the following lemmas.

\begin{lemma}\label{lem:gradW}
    Given the error dynamics \eqref{er-2} and the error kinematics \eqref{error-dynamicsm1}, the time-evolution of the gradient of the system states with respect to the NN weights in $\hat W$ satisfy the following systems of matrix differential equations:
    \begin{align}\label{eq:gradwg}
       \frac{\partial}{\partial t} \nabla^L_{\breve{\hat{W}}}\tilde g &= -\ad_{\tilde\xi} (\nabla^L_{\breve{\hat{W}}}\tilde g)+\nabla_{\breve{\hat{W}}}\tilde\xi\in\R^{n\times nm},\\\nonumber
       \frac{\partial}{\partial t} \nabla_{\breve{\hat{W}}}\tilde \xi &= -A \nabla_{\breve{\hat{W}}}\tilde \xi-k_p(\Hes^L_{\tilde g}\psi)\nabla^L_{\breve{\hat{W}}}\tilde g\\
       &+\left[
           \sigma_1\bB ~\cdots~\sigma_m\bB
       \right]_{y=\hat Vx}\in\R^{n\times nm}.\label{eq:gradwxi}
    \end{align}
\end{lemma}
\begin{proof}
    The proof is presented in Appendix \ref{app:lem3&4}.
\end{proof}

\begin{lemma}\label{lem:gradV}
    Given the error dynamics \eqref{er-2} and the error kinematics \eqref{error-dynamicsm1}, the time-evolution of the gradient of the system states with respect to the NN weights in $\hat V$ satisfy the following system of matrix differential equations:
   \begin{align}
       \frac{\partial}{\partial t} \nabla^L_{\breve{\hat{V}}}\tilde g &= -\ad_{\tilde\xi} (\nabla^L_{\breve{\hat{V}}}\tilde g)+\nabla_{\breve{\hat{V}}}\tilde\xi\in\R^{n\times mM},\\\nonumber
       \frac{\partial}{\partial t} \nabla_{\breve{\hat{V}}}\tilde \xi &= -A \nabla_{\breve{\hat{V}}}\tilde \xi-k_p(\Hes^L_{\tilde g}\psi)\nabla^L_{\breve{\hat{V}}}\tilde g\\ \label{eq:gradxi}
       &+\bB\hat W\left[
           x_{1}\Pi ~\cdots~x_{M}\Pi
       \right]\in\R^{n\times mM},
    \end{align}
    where 
$$\Pi\coloneqq\left.\nabla_y\sigma\right|_{y=\hat Vx}=\1_{m\times m}-\begin{bmatrix}
    \sigma_1^2 & \0 & \0\\
    \0 & \ddots & \0\\
    \0 & \0 & \sigma_m^2
\end{bmatrix}_{y=\hat Vx}.$$
\end{lemma}
\begin{proof}
    The proof is presented in Appendix \ref{app:lem3&4}.
\end{proof}

 We use the following cost function to propose appropriate learning rules for training the NN weights:
\begin{align}
  F(\tilde g,\tilde\xi,t,\hat W,\hat V) &=   \frac{\alpha_{1}}{2} (\tilde{\xi}-\tilde{\xi}^*)^T A (\tilde{\xi}-\tilde{\xi}^*)+ \alpha_{2} \psi((\tilde{g}^*)^{-1}\tilde g), \label{obj}
\end{align}
This function is a weighted summation of the norm square of the difference between the actual and nominal velocity errors and the error function applied to the difference between the actual and nominal configuration errors. Here, $\alpha_i>0$ $(i=1,2)$ is the associated weight to each term.  
Note that $F$ is indirectly a function of the approximated weights $\hat W$ and $\hat V$ through defining $u$ based on \eqref{opt-con2} and considering the state trajectories as a variation of the system evolution under the ideal controller. In addition, this function is designed to ensure that the trajectory errors converge to zero in the same manner as in the nominal system.
Therefore, we define the learning rules by:
\begin{align}
  \dot{\hat{W}} &= -\rho_{1} \nabla_{\hat W} F-\gamma_{1} \bigg(\|\tilde{\xi}-\tilde{\xi}^*\|\left(1+k_p\|\Hes_{\tilde g}^L\psi\|\|\nabla_{\breve{\hat{W}}}^L\tilde g\|\right) \nonumber\\
  &+\|\der^L_{(\tilde{g}^*)^{-1}\tilde g} \psi\|\|\nabla_{\breve{\hat{W}}}^L\tilde g\|\bigg) \hat{W},  \label{w1}\\
  \dot{\hat{V}}&=-\rho_2 \Pi\hat{W}^T\hat{W}\sigma(\hat{V}x)x^T-\gamma_{2} \|\tilde{\xi}\| \hat{V}.\label{v3}
\end{align}
Here, $\rho_1,\rho_2>0$ and $\gamma_1,\gamma_2>0$ are learning parameter.

\begin{lemma}\label{lem:gradu}
    The gradient of the norm of the control effort with respect to $\hat{V}$ is found by
    \begin{align}
        \nabla_{\hat{V}}\Big(\frac{1}{2}u^T u\Big)=\Pi\hat{W}^T\hat{W}\sigma(\hat{V}x)x^T.
    \end{align}
\end{lemma}
\begin{proof}
    The proof is presented in Appendix \ref{app:lem5}.
\end{proof}
\begin{remark}
    The learning rule in \eqref{w1}-\eqref{v3} is based on the popular gradient decent adaptation rule with dissipation. It should be noted that \eqref{v3} is indeed based on the gradient descent rule considering only the norm square of the control signal (see Lemma \ref{lem:gradu}). The learning rule is then designed in a manner that finding optimal $\hat W$ will enhance the control performance, whilst the optimum $\hat V$ will help reduce the amount of control effort.
\end{remark}
\begin{assumption}\label{as:staticapprox} (Static Approximation) 
    We assume that the rate of convergence of the NN weights $\hat{W}$ and $\hat{V}$ is much higher than that of the controller, i.e., we assume that the following static approximation conditions hold:
    \begin{align}
         \frac{\partial}{\partial t} \nabla_{\breve{\hat{W}}}\tilde \xi\equiv \0_{n\times nm} ~~~~~\text{\&}~~~~~\frac{\partial}{\partial t} \nabla_{\breve{\hat{V}}}\tilde \xi\equiv \0_{n\times mM}.
    \end{align}
\end{assumption}
\begin{assumption}
\label{as:ApproxInertia}
    To have well-defined learning rules, we assume knowledge of a positive lower bound for the minimum eigenvalue of the time-varying matrix $\bB^{-1}$, i.e., $0<\varsigma\leq\lambda_{min}(\bB^{-1})$. Then we use $1/\varsigma\1_{n\times n}$ instead of $\bB$ to evaluate the adaptation laws for the NN controller.
\end{assumption}
Based on Lemmas \ref{lem:gradW} and \ref{lem:gradV}, we explicitly calculate the learning rule \eqref{w1} in the following proposition.
\begin{proposition}
    Given Assumptions \ref{as:staticapprox} and \ref{as:ApproxInertia} and the cost function $F$ in \eqref{obj}, the learning rule \eqref{w1} reads:
    \begin{align}\nonumber
        \dot{\hat{W}}&=-\frac{\rho_1\alpha_1}{\varsigma}(\tilde\xi-\tilde\xi^*)\sigma^T +k_p\rho_1\alpha_1\left[\left(\left(\Hes_{\tilde g}^L\psi\right)\nabla^L_{\breve{\hat{W}}}\tilde g \right)^T(\tilde\xi-\tilde\xi^*)\right]\breve\breve\\
        &- \rho_1\alpha_2\left[\left(\left(\der_{(\tilde g^*)^{-1}\tilde g}^L\psi\right)\nabla^L_{\breve{\hat{W}}}\tilde g\right)^T\right]\breve\breve\nonumber
        -\gamma_{1} \bigg(\|\der^L_{(\tilde g^*)^{-1}\tilde g}\psi\| \|\nabla_{\breve{\hat{W}}}^L\tilde g\|\\ &+\|\tilde\xi-\tilde\xi^*\|\left(1+k_p\|\Hes_{\tilde g}^L\psi\|\|\nabla_{\breve{\hat{W}}}\tilde g\|\right)\bigg) \hat{W},\label{w3}
    \end{align}
    where at every $(\tilde g, \tilde\xi)$ and NN input $x\in\R^M$, $\nabla^L_{\breve{\hat{W}}}\tilde g$ is found by solving the matrix differential equation
    \begin{align}\nonumber
        \frac{\partial}{\partial t} \nabla^L_{\breve{\hat{W}}}\tilde g = &-\ad_{\tilde\xi} (\nabla^L_{\breve{\hat{W}}}\tilde g)-A^{-1}\Big(k_p(\Hes_{\tilde g}^L\psi)\nabla^L_{\breve{\hat{W}}}\tilde g\\ &-\left[
           \frac{\sigma_1}{\varsigma}\1_{n\times n} ~\cdots~\frac{\sigma_m}{\varsigma}\1_{n\times n}
       \right]_{y=\hat Vx}\Big).\label{eq:gradwg1}
    \end{align}
\end{proposition}
\begin{proof}
    The proof is presented in Appendix \ref{app:prop2}.
\end{proof}

Now that we have introduced the learning rules, the control loop is closed by feeding back the error states and feeding forward the desired trajectory and its derivatives, as required to approximate the ideal controller \eqref{opt-cont1}. That is, we define 
\begin{align}\label{xx}
    x= [\breve{\tilde{g}}^T~\tilde\xi^T~\breve{{g}}_d^T~\xi_d^T~\dot\xi_d^T]^T.
\end{align}
By adding and subtracting $\bB W \sigma(\hat{V} x)$ in \eqref{er-2} and defining $\tilde W\coloneqq W-\hat W$, then the closed-loop system can be stated as
\begin{eqnarray}
  \dot{\tilde{\xi}} =  -A \tilde{\xi}-k_p(\der^L_{\tilde g}\psi)^T-\bB\tilde{W} \sigma(\hat{V} x)+w\label{er-3}
\end{eqnarray}
where  $w= \bB[W(\sigma(\hat V x)-\sigma({V} x))+\cE]$. Considering the Universal Approximation property of NNs, the ideal weights of the NN, $\sigma$ in \eqref{eq:sigma} as the activation function, and NN approximation error are all bounded. Therefore, there exists an unknown positive upper bound such that $\|w\|\leq \bar{w}$.
%
\begin{thm}\label{thm1}
  Consider the uncertain and disturbed Euler-Poincar\'{e} system \eqref{dynamicsm1}-\eqref{dynamicsm2} on the matrix Lie group $G$ that follows a bounded smooth desired trajectory $t\to g_d(t)\in G$. Let $\psi$ be a locally quadratic configuration error function for the system with constants $\Theta,b_1,b_2>0$. The controller \eqref{opt-con2} with input states defined in \eqref{xx} and learning rules \eqref{v3} \& \eqref{w3}-\eqref{eq:gradwg1} guarantees that the weight error $\tilde W$ (and hence $\hat W$) is globally ultimately bounded with the bound 
  \begin{align}
      \!r_1\!\!\!= \max\{\frac{\rho_1\alpha_1 \sqrt{m}}{\varsigma\gamma_1}, \frac{\rho_1\alpha_1 }{\gamma_1},   \frac{\rho_1\alpha_2}{\gamma_1}  \}+\varpi>0,
  \end{align}
  where $\varpi\geq\|W\|$ is the bound for the ideal NN weights.
Further, let us define the following constants: 
\begin{align*}
a_1&\coloneqq\lambda_{min}(A),\\
a_2&\coloneqq\lambda_{max}(A),\\
\bar\upsilon&\coloneqq\sqrt{\frac{b_2(1+2b_1)}{\theta_0}}(\frac{\sqrt{m}r_1}{\varsigma}+\bar w),~~~~~0<\theta_0<\Theta\\
    h&\coloneqq\!\!\!\!\!\!\sup_{\tilde g\in\psi^{-1}(\leq\theta_0)}\!\!\!\|\Hes_{\tilde g}^L\psi\|<+\infty,\\
    \bar\epsilon&\coloneqq\min\{\frac{a_1}{h+a_2^2/4},\sqrt{2b_1},\frac{a_1-\bar\upsilon}{h+a_2/2}\}.
\end{align*}
If the system's initial condition satisfies $\psi(\tilde g(0))<\theta_0$ and the control parameters are designed such that $a_2\geq a_1>\bar\upsilon$ and $k_p>\frac{a_1}{\bar\epsilon}-h$, then the configuration and velocity errors $(\tilde g,\tilde\xi)$, as well as the approximated weight matrix $\hat V$ remain uniformly ultimately bounded.
\end{thm}
\begin{proof}

The proof follows $3$ steps to show the ultimate boundedness of (i) the weight error $\tilde W$, (ii) the configuration and velocity error $(\tilde g,\tilde\xi)$, and (iii) the weight error $\tilde V\coloneqq V-\hat{V}$:

\textbf{Step (i)}
To guarantee the ultimate boundedness of $\tilde W$, the following Lyapunov  candidate is considered
\begin{eqnarray}
  L_1 (\tilde W)=  \frac{1}{2}\| \tilde{W}\| ^2.\label{lyap1}
\end{eqnarray}
Taking the time derivative of $L_1$ along the direction of the vector field defining the closed-loop dynamics considering the learning rules, leads to
\begin{align}
  \dot{L}_1= \tr(\dot{\tilde{W}}^T \tilde{W}). \label{L1dp1}
\end{align}
Since $\dot{\tilde{W}}=-\dot{\hat{W}}$ and by substituting (\ref{w3}) to (\ref{L1dp1}) and defining $\bg\coloneqq(\tilde g^*)^{-1}\tilde g$ and $\Xi\coloneqq\tilde\xi-\tilde\xi^*$  we have
 \begin{align}\nonumber
 &\tr(\dot{\tilde{W}}^T \tilde{W}) = \tr\Bigg(\bigg(\rho_1\alpha_2\left[\left(\left(\der_{\bg}^L\psi\right)\nabla^L_{\breve{\hat{W}}}\tilde g\right)^T\right]\breve\breve\\
        &~~~+\frac{\rho_1\alpha_1}{\varsigma}\Xi(\sigma(\hat{V}x))^T -k_p \rho_1\alpha_1\left[\left(\left(\Hes_{\tilde g}^L\psi\right)\nabla^L_{\breve{\hat{W}}}\tilde g \right)^T\!\!\!\Xi\right]\breve\breve\nonumber\\
        &+\gamma_{1} \bigg(\|\Xi\|\left(1+k_p\|\Hes_{\tilde g}^L\psi\|\|\nabla_{\breve{\hat{W}}}\tilde g\|\right)+\|\der^L_{\bg} \psi\|\|\nabla_{\breve{\hat{W}}}^L\tilde g\|\bigg) \hat{W}\bigg)^{\!\!\! T} \tilde{W}\!\!\Bigg) \nonumber
  \end{align}
Now, by considering that $\sigma_i^2(.)\leq 1$ ($i=1,\ldots,m$), we lead to the following inequalities
\begin{align*}
 &\tr\Bigg(\tilde{W}^T\bigg(\frac{\rho_1\alpha_1}{\varsigma}\Xi(\sigma(\hat{V}x))^T\!\!\!\!-\!k_p\rho_1\alpha_1\left[\left(\left(\Hes_{\tilde g}^L\psi\right)\nabla^L_{\breve{\hat{W}}}\tilde g \right)^T\!\!\!\Xi\right]\breve\breve \\ &+ \rho_1\alpha_2\left[\left(\left(\der_{\bg}^L\psi\right)\nabla^L_{\breve{\hat{W}}}\tilde g\right)^T\right]\breve\breve\bigg)\Bigg) \leq  \bigg(\rho_1\alpha_2\|\der_{\bg}^L\psi\|\| \nabla^L_{\breve{\hat{W}}}\tilde g\|\\ &~~~~~~~~~~~~~~~~+\rho_1\alpha_1\|\Xi\|\big(\frac{\sqrt{m}}{\varsigma}+k_p\|\Hes_{\tilde g}^L\psi\|\|\nabla^L_{\breve{\hat{W}}}\tilde g\|\big)\bigg)\|\tilde W\|, \nonumber
        \end{align*}
        and
 \begin{align*}
& \tr\Bigg(\gamma_{1} \bigg(\|\Xi\|\left(1+k_p\|\Hes_{\tilde g}^L\psi\|\|\nabla_{\breve{\hat{W}}}\tilde g\|\right)+\|\der^L_{\bg} \psi\|\|\nabla_{\breve{\hat{W}}}^L\tilde g\|\bigg) \hat{W}^T\tilde{W}\Bigg)\\ &\leq \gamma_{1} \bigg(\|\Xi\|\left(1+k_p\|\Hes_{\tilde g}^L\psi\|\|\nabla_{\breve{\hat{W}}}\tilde g\|\right)\\ 
&~~~~~~~~~~~~~~~~~~~~~~~~~+\|\der^L_{\bg} \psi\|\|\nabla_{\breve{\hat{W}}}^L\tilde g\| \bigg)\left(\varpi\|\tilde W\|-\|\tilde W\|^2\right). \nonumber
   \end{align*}
Therefore, by rearranging the terms, we have
\begin{align*}
  \dot{L}_1 &\leq \bigg(\frac{\rho_1\alpha_1\sqrt{m}}{\varsigma}+ \varpi \gamma_{1} \bigg)\|\Xi\|\|\tilde W\|\\
  &+k_p\bigg(\rho_1\alpha_1+ \varpi \gamma_{1} \bigg)\|\Hes_{\tilde g}^L\psi\|\|\nabla_{\breve{\hat{W}}}\tilde g\|\|\Xi\|\|\tilde W\|\\
 & +\bigg(\rho_1\alpha_2+\varpi\gamma_1\bigg)\|\der_{\bg}^L\psi\|\| \nabla^L_{\breve{\hat{W}}}\tilde g\|\|\tilde W\|-\gamma_{1} \|\Xi\|\|\tilde W\|^2\\
 &-\gamma_1\bigg(k_p\|\Xi\|\|\Hes_{\tilde g}^L\psi\|
 +\|\der^L_{\bg} \psi\|\bigg)\|\nabla_{\breve{\hat{W}}}^L\tilde g\|\|\tilde W\|^2
\end{align*}
By defining  $ \kappa=\frac{\rho_1\alpha_1 \sqrt{m}/\varsigma+ \varpi \gamma_{1}}{2  \sqrt{\gamma_{1}}}>0$,  and adding and subtracting $\kappa^2 \|\Xi\|$  to the above inequality, one can get
\begin{align*}
  \dot{L}_1 &\leq \|\Xi\|\bigg(\kappa^2-\big(\sqrt{\gamma_1}\|\tilde W\|-\kappa\big)^2\bigg) \\
 &+k_p\bigg(\rho_1\alpha_1+ \varpi \gamma_{1}-\gamma_1\|\tilde W\| \bigg)\|\Hes_{\tilde g}^L\psi\|\|\nabla_{\breve{\hat{W}}}\tilde g\|\|\Xi\|\|\tilde W\|\\ 
 & +\bigg(\rho_1\alpha_2+\varpi\gamma_1-\gamma_1\|\tilde W\|\bigg)\|\der_{\bg}^L\psi\|\| \nabla^L_{\breve{\hat{W}}}\tilde g\|\|\tilde W\|. 
\end{align*}
Therefore, if
\begin{align}
    \|\tilde{W}\|\! \geq \!r_1\!\!\!=\beta+\varpi\coloneqq \max\{\frac{\rho_1\alpha_1 \sqrt{m}}{\varsigma\gamma_1}, \frac{\rho_1\alpha_1 }{\gamma_1},   \frac{\rho_1\alpha_2}{\gamma_1}  \}+\varpi>0, \nonumber
\end{align}
$\dot{L}_1<0$, i.e. it is negative definite on and outside the ball with radius $r_1$ in the space of NN weight estimation error. Note that negative definiteness of $\dot L_1$ on the ball is guaranteed unless $\frac{\alpha_1 \sqrt{m}}{\varsigma}= \alpha_1 =   \alpha_2$, which can be easily avoided by modifying the learning parameters.  This proves the boundedness of $\tilde W$ with the ultimate bound $r_1$, i.e., in finite time $T_1$, $\|\tilde W\|$ becomes less than or equal to $r_1$ and remains such, for $t>T_1$. Since the bound $r_1>\varpi$, if the initial condition $\hat W(0)$ satisfies the inequality $\|\hat W(0)\|<\beta$, then $\|\tilde W(t)\|< r_1$ for all $t$ and $T_1=0$. 

\textbf{Step (ii)} To guarantee the ultimate boundedness of $(\tilde g,\tilde \xi)$, the following Lyapunov function is defined for a $\delta>0$:
\begin{align}
      L_2 (\tilde g,\tilde\xi)&= \frac{1}{2} \tilde{\xi}^T \tilde{\xi} + k_p\psi(\tilde g) + \delta\dot{\psi}(\tilde g,\tilde\xi).\label{lyap2}
   \end{align}
Under Assumption \ref{as:quad} for a $\theta_0<\Theta$, we first show that this is indeed a positive-definite function on the set $\Omega=\{(\tilde g,\tilde\xi)|\psi(\tilde g)\leq\theta_0\}$, when $\delta<\sqrt{2b_1}$:
\begin{align*}
    L_2&=\frac{1}{2}\|\tilde\xi\|^2+k_p\psi+\delta(\der^L_{\tilde g}\psi)\tilde\xi\\ &\geq \frac{1}{2}\|\tilde\xi\|^2+k_pb_1\|\der^L_{\tilde g}\psi\|^2-\delta\|\der^L_{\tilde g}\psi\|\|\tilde\xi\|\\
    &\geq \frac{1}{2}\|\tilde\xi\|^2+b_1\|\der^L_{\tilde g}\psi\|^2-\delta\|\der^L_{\tilde g}\psi\|\|\tilde\xi\|\\
    &= \frac{1}{2}(\|\tilde\xi\|-\delta\|\der^L_{\tilde g}\psi\|)^2+(b_1-\frac{\delta^2}{2})\|\der^L_{\tilde g}\psi\|^2\geq 0.
\end{align*}
The third line is true since $k_p\geq 1$. In addition, based on \eqref{eq:b1b2}, we have the norm inequalities:
\begin{align}\nonumber
   & c_1\left\|\begin{bmatrix}
        (\der^L_{\tilde g}\psi)^T\\ \tilde\xi    \end{bmatrix}\right\|^2\leq\begin{bmatrix}
        \|\der^L_{\tilde g}\psi\|\\ \|\tilde\xi\|    \end{bmatrix}^T\overbrace{\begin{bmatrix}
       k_pb_1  & \frac{-\delta}{2}\\ \frac{-\delta}{2} & \frac{1}{2}    \end{bmatrix}}^{=C_1}\begin{bmatrix}
        \|\der^L_{\tilde g}\psi\|\\ \|\tilde\xi\|
    \end{bmatrix}\leq L_2\\
    & \leq \begin{bmatrix}
        \|\der^L_{\tilde g}\psi\|\\ \|\tilde\xi\|    \end{bmatrix}^T\overbrace{\begin{bmatrix}
        k_pb_2  & \frac{\delta}{2}\\ \frac{\delta}{2} & \frac{1}{2}   \end{bmatrix}}^{=C_2}\begin{bmatrix}
        \|\der^L_{\tilde g}\psi\|\\ \|\tilde\xi\|
    \end{bmatrix}\leq c_2\left\|\begin{bmatrix}
        (\der^L_{\tilde g}\psi)^T\\ \tilde\xi    \end{bmatrix}\right\|^2\label{eq:quadL2}
\end{align}
for positive-definite matrices $C_1$ and $C_2$, with $c_1=\lambda_{min}(C_1)$ and $c_2=\lambda_{max}(C_2)$. We also define the constants $a_1\coloneqq\lambda_{min}(A)$, $a_2\coloneqq\lambda_{max}(A)$, and 
\begin{align*}
    h\coloneqq\!\!\!\!\!\!\sup_{\tilde g\in\psi^{-1}(\leq\theta_0)}\!\!\!\|\Hes_{\tilde g}^L\psi\|<+\infty,
\end{align*}
which is well-defined on the compact set $\psi^{-1}(\leq\theta_0)$ ($\psi$ is proper), since $\psi$, and consequently its derivatives, are smooth on $G$.
These constants are used in the proof of boundedness of the states.

Taking the derivative of $L_2$ along the direction of the closed-loop vector field leads to 
    \begin{align*}
        \dot{L}_2\!\!&=\!\!k_p(\der_{\tilde g}^L\psi) \tilde{\xi}+\tilde\xi^T\dot{\tilde\xi}+\delta\tilde\xi^T(\Hes_{\tilde g}^L\psi) \tilde{\xi}+\delta(\der_{\tilde g}^L\psi) \dot{\tilde{\xi}}=k_p(\der_{\tilde g}^L\psi) \tilde{\xi}\\ &+ (\tilde\xi^T\!\!\!\!+\delta\der_{\tilde g}^L\psi)\big(-A\tilde\xi-k_p(\der_{\tilde g}^L\psi)^T-\bB\tilde W\sigma(\hat Vx)+w\big)\\ &+\delta\tilde\xi^T(\Hes_{\tilde g}^L\psi) \tilde{\xi}\\
        &\leq-a_1\|\tilde\xi\|^2-k_p\delta\|\der_{\tilde g}^L\psi\|^2+ \delta a_2\|\tilde\xi\|\|\der_{\tilde g}^L\psi\|
        +\delta h\|\tilde\xi\|^2\\
        &+\begin{bmatrix}
            \delta D^T & D^T
        \end{bmatrix}\begin{bmatrix}
        (\der^L_{\tilde g}\psi)^T\\ \tilde\xi    \end{bmatrix},
    \end{align*}
    such that $D=-\bB\tilde W\sigma(\hat Vx)+w$ and it satisfies $\|D\|\leq\frac{\sqrt{m}r_1}{\varsigma}+\bar w$ assuming that $\|\hat{W}(0)\|<\beta$ ($T_1=0$). Therefore,
    \begin{align*}
        \dot L_2&\leq -\begin{bmatrix}
        \|\der^L_{\tilde g}\psi\|\\ \|\tilde\xi\|    \end{bmatrix}^T\overbrace{\begin{bmatrix}
            k_p\delta & \frac{-\delta a_2}{2}\\\frac{-\delta a_2}{2} & a_1-\delta h
        \end{bmatrix}}^{=E}\begin{bmatrix}
        \|\der^L_{\tilde g}\psi\|\\ \|\tilde\xi\|    \end{bmatrix} \\ &+ \sqrt{1+\delta^2}(\frac{\sqrt{m}r_1}{\varsigma}+\bar w)\left\|\begin{bmatrix}
        \der^L_{\tilde g}\psi)^T\\ \tilde\xi    \end{bmatrix}\right\|.
    \end{align*}
    Matrix $E\in\R^{2\times 2}$ is positive-definite if and only if its trace and determinant are both positive. The determinant 
    \begin{align*}
        \det(E)=k_p\delta(a_1-\delta h)-\frac{\delta^2 a_2^2}{4}>0 \iff
        \delta<\frac{k_pa_1}{k_ph+a_2^2/4},
        \end{align*}
        and if $\delta<a_1/h$, the trace $\tr(E)=k_p\delta+a_1-\delta h>k_p\delta>0$.        
    Since for $k_p\geq 1$ we have
    \begin{align*}
        \frac{a_1}{h+a_2^2/4}<\frac{k_pa_1}{k_ph+a_2^2/4}<\frac{a_1}{h},
    \end{align*}
   it is concluded that there exists a $\delta<\epsilon\coloneqq\min\{\frac{a_1}{h+a_2^2/4},\sqrt{2b_1}\}$ that guarantees $L_2$ is a well-defined Lyapunov function with negative-definite time derivative if the closed-loop dynamics \eqref{er-3} is unperturbed ($\tilde W=\0_{n\times m}$ and $w=\0_{n\times 1}$).     
    Now, let us define $e_1\coloneqq \lambda_{min}(E)$ and $\upsilon\coloneqq\sqrt{1+2b_1}(\frac{\sqrt{m}r_1}{\varsigma}+\bar w)$. Then,
    \begin{align*}
        \dot L_2< -e_1\left\|\begin{bmatrix}
        (\der^L_{\tilde g}\psi)^T\\ \tilde\xi   \end{bmatrix}\right\|^2+\upsilon \left\|\begin{bmatrix}
        (\der^L_{\tilde g}\psi)^T\\ \tilde\xi    \end{bmatrix}\right\|
    \end{align*}
   is negative-definite on and outside the ball $\cB_{\upsilon/e_1}=\{X\in\R^{2n}|~\|X\|\leq \upsilon/e_1\}$. However, this ball must be contained in the set $\bar\Omega=\{[\der^L_{\tilde g}\psi~~\tilde\xi^T]^T\in\R^{2n}|~\psi(\tilde g)\leq\theta_0\}$ to guarantee that $\psi$ remains quadratic or \eqref{eq:b1b2} holds, which is instrumental to defining $L_2$ and performing stability analysis. 
   From the right inequality in \eqref{eq:b1b2}, we have that if
   \begin{align*}
      b_2\|\der^L_{\tilde g}\psi\|^2\leq \theta_0 \iff \|\der^L_{\tilde g}\psi\|\leq \sqrt{\frac{\theta_0}{b_2}},
   \end{align*}
   then $\psi(\tilde g)\leq \theta_0$, i.e., $[\der^L_{\tilde g}\psi~~\tilde\xi^T]^T\in\bar\Omega$ or $(\tilde g,\tilde\xi)\in\Omega$.
Accordingly, if $e_1>\bar\upsilon\coloneqq\upsilon\sqrt{\frac{b_2}{\theta_0}}$, then 
$\cB_{\upsilon/e_1}\subset\bar\Omega$. Note that on the boundary of $\cB_{\upsilon/e_1}$, $\dot L_2$ is negative.

Let us define $\bar\epsilon\coloneqq\min\{\epsilon,\frac{a_1-\bar\upsilon}{h+a_2/2}\}$.
Then, for all $a_2\geq a_1>\bar\upsilon$ and $k_p>\frac{a_1}{\bar\epsilon}-h$, there exists a $\frac{a_1}{k_p+h}<\delta<\bar\epsilon$, such that
\begin{align*}
    e_1&=\frac{k_p\delta+a_1-\delta h-\sqrt{(a_1-\delta h-k_p\delta)^2+\delta^2a_2^2}}{2}\\
    &>\frac{k_p\delta+a_1-\delta h-|a_1-\delta h-k_p\delta|-\delta a_2}{2}\\
    &=a_1-\delta (h+a_2/2)>\bar\upsilon.
\end{align*}
Hence, for any error trajectory of the system with the initial conditions satisfying $[\der^L_{\tilde g(0)}\psi~~\tilde\xi(0)^T]^T\in\bar\Omega\diagdown\cB_{\upsilon/e_1}$, the Lyapunov function $L_2$ and $\|[\der^L_{\tilde g}\psi~~\tilde\xi^T]\|$ (according to \eqref{eq:quadL2}) exponentially decrease for a finite time $T_2$ until it reaches the ball $\cB_{\upsilon/e_1}\subset\bar\Omega$ and stays there afterward. This proves the boundedness of $L_2$ with an ultimate bound $\frac{c_2\upsilon^2}{e_1^2}$, which also gives the ultimate bounds $r_2\coloneqq \frac{b_2\upsilon^2}{e_1^2}$ and $r_3\coloneqq \frac{\upsilon}{e_1}$ for $\psi(\tilde g)$ and $\tilde\xi$, respectively.

\textbf{Step (iii)} To complete the stability analysis, based on (\ref{v3}), let us describe the behaviour of the hidden layer weight estimation error $\tilde{V}\coloneqq V-\hat{V}$ as 
\begin{align}
    \dot{\tilde{V}}&=-\dot{\hat{V}}=\rho_2  \Pi^T\hat{W}^T\hat{W}\sigma(\hat Vx)x^T\nonumber
        +\gamma_{2} \|\tilde{\xi}\|(V-\tilde{V}) \nonumber\\
        &= -\mathfrak{A} \tilde{V}+\mathfrak{F} ,\label{ev1}
\end{align}
where $\mathfrak{A}\coloneqq\gamma_{2} \|\tilde{\xi}\|, \mathfrak{F}=\rho_2 \Pi^T\hat{W}^T\hat{W}\sigma(Vx)x^T    
+\gamma_{2} \|\tilde{\xi}\| V$.         
Boundedness of the error signals ($\tilde g$ and $\tilde\xi$) and the desired trajectory and its derivatives ($g_d$, $\xi_d$, and $\dot\xi_d$) leads to the boundedness of the NN input $x$. Further, since 
$\hat{W}$, $\tilde{\xi}$, and $x$ are all bounded as well as the ideal weights $V$ and $\sigma$ function, we conclude the boundedness of $\mathfrak{F}$ and $\mathfrak{A}$. Therefore, (\ref{ev1}) is a first-order equation with bounded coefficients. $\mathfrak{A}$ is positive; hence, $\tilde{V}$ will be also ultimately bounded and the proof is complete.
\end{proof}
\begin{remark}
  Several parameters are in the learning rules (\ref{v3}) and (\ref{w3})-\eqref{eq:gradwg1} which affect the convergence behaviour of the NN. Increasing the learning rates $\rho_{1,2}$ leads to faster convergence. However, it may lead to more oscillations in transient response or narrow cost functions. The parameters $\alpha_{1,2}$ represent the weight of each term in the objective function (\ref{obj}). Increasing each of them yields a higher contribution of the corresponding term in the weight adjustment. For example, increasing  $\alpha_{1}$ relative to $\alpha_{2}$ results in smaller velocity error. The parameters $\gamma_{1,2}$ are damping factors of the modification terms and play an important role in providing stability of the overall system. On the other hand, if $\gamma_{1,2}$ are too large, we can lead to a saturation. 
    \end{remark}
    \begin{remark}
        The cost function $F$ in \eqref{obj} may be defined for the trajectory error dynamics, instead of its difference from the nominal error dynamics. In this case, the stability result presented in Theorem \ref{thm1} would hold; however, it leads to the poor performance of the controller (large control action and overshoot) in the transient phase when the system has large initial errors.
    \end{remark}
    \begin{remark}
        The ultimate bound for the NN weight estimation error $\tilde W$ is governed by the learning parameters, and it may not be reduced below $\varpi$. Further, the ultimate bounds for the error states $(\tilde g,\tilde\xi)$ can be directly controlled by the parameters $b_{1,2}, \Theta$, corresponding to the design of the configuration error function $\psi$, the control parameter $k_p$, and the eigenvalues of $A$. The bounds are obtained under some inequality conditions on the control parameters. This can also be translated as given a set of control parameters, the neural network controller can attenuate certain levels of disturbances and uncertainties in the system. These inequalities are derived according to the fact that the configuration error function $\psi$ is only locally quadratic, in general. Under the stronger assumption of $\psi$ being a global configuration error function that is globally quadratic (Definition \ref{def:GlobConfErr}), the neural network controller provides globally ultimately bounded stability of the system for all choices of control parameters (see Corollary \ref{cor1}). Note that for systems whose configuration manifold $G$ is compact, there does not exist any globally quadratic $\psi$, since any positive-definite smooth function obtains at least a maximum on $G$.
    \end{remark}

    \begin{corollary}\label{cor1}
Consider the uncertain and disturbed Euler-Poincar\'{e} system \eqref{dynamicsm1}-\eqref{dynamicsm2} on the matrix Lie group $G$ equipped with a global configuration error function $\psi$ with the globally quadratic property that follows a bounded smooth desired trajectory $t\to g_d(t)\in G$. The controller \eqref{opt-con2} with input states defined in \eqref{xx} and learning rules \eqref{v3} \& \eqref{w3}-\eqref{eq:gradwg1} guarantees that the weight error $\tilde W$ is globally ultimately bounded. 
Then, for all initial conditions $(g(0),\xi(0))\in G\times\R^n$ the configuration and velocity errors $(\tilde g,\tilde\xi)$, as well as the approximated weight matrix $\hat V$ remain globally uniformly ultimately bounded.

    \end{corollary}
    \begin{proof}
        The proof follows the same lines of proof of Theorem \ref{thm1}, except since $\psi$ is globally quadratic we can always find a $\delta$ that ensures $L_2$ is positive definite and ultimately bounded, with no conditions on the control gains.
    \end{proof}

\section{Neural Network Formation Control of Unknown Multi-Agents: A Case Study} \label{NNcont_mtl}
\subsection{Decentralized Formation Control of Multi-agents on $\SE$}

As a case study, let us consider a  team of $N_v$ vehicles (3D rigid bodies) on the Special Euclidean group $\SE$, labeled by $l=1,\ldots,N_v$, and one virtual leader indexed by $l=0$. For a quick review on $\SE$ we refer the reader to Appendix \ref{app:SE(3)}. 
 Let us denote the pose and the body velocity of the virtual leader in a reference frame by $g_0\in\SE$ and $\xi_0\in\R^6$, such that $\xi_0^\wedge\in\se$. We consider a desired trajectory $t\to \bar{g}_{0}(t)\in\SE$ for the leader with the desired body velocity $\bar{\xi}_{0}$.
Hence, the leader equations of motion are
\begin{align}
  \dot{g}_0 &= g_0  {\xi}_0^\wedge\in\R^{4\times 4},\label{leaderkinematics} \\
  \dot\xi_0 &= \bI_0^{-1}\ad^T_{\xi_0}\bI_0\xi_0+\bI_0^{-1}u_0\in\R^6,\label{leaderdynamics}
\end{align}
where $\I_0\in\R^{6\times 6}$ is the leader's inertia matrix, and since the leader dynamics is known and undisturbed, its control input $u_0$ can be determined based on the control law \eqref{ustar1} or any other controller to track the desired trajectory.  
Similarly for the vehicle $l\in\{1,\ldots,N_v\}$, we denote its pose and body velocity by $g_l\in\SE$ and $\xi_l\in\R^6$ and write its dynamics based on \eqref{dynamicsm1}-\eqref{dynamicsm2}: 
 \begin{align}
  \dot g_l &=g_l\xi_l^\wedge\in\R^{4\times 4},\label{agentkinematics}\\
  \dot\xi_l &= \bI_l^{-1}\ad^T_{\xi_l}\bI_l\xi_l+\bB_l u_l+\mu_l(g_l,\xi_l)+d_l\in\R^6.
  \label{agentdynamics}
\end{align}
Here, $\I_l\in\R^{6\times 6}$ is the $l$th agent's unknown inertia. We assume the lower bound $\varsigma_l$ for the minimum eigenvalue of $\bB^{-1}$, $u_l$ is the control input (including the torque and force about its center of mass), and $\mu_l$ and $d_l$ refer to bounded uncertainties and disturbances for agent $l$.
 We now define the relative pose between the $l$th agent and the virtual leader by $g_{l0}:=g_0^{-1}g_l\in\SE$.
\begin{problem}[Decentralized Formation Control]
Let us consider a multi-agent system with the agents' uncertain and disturbed dynamics \eqref{agentkinematics}-\eqref{agentdynamics} and the virtual leader dynamics \eqref{leaderkinematics}-\eqref{leaderdynamics} tracking a smooth trajectory $t\to \bar g_{0}(t)$. 
Given a formation, i.e., reference poses of the agents relative to the virtual leader $\bar g_{l0}\in \SE$ ($l=1,\ldots N_v$), find the local control signals $u_l$ based on \eqref{opt-con2} to make and keep the formation. The control strategy should be decentralized and only use the local information of an individual agent and the known states of the virtual leader. 
\end{problem}
Assuming the availability of the virtual leader states to all agents by solving \eqref{leaderkinematics}-\eqref{leaderdynamics}, we redefine this problem as a trajectory-tracking problem for each agent. The desired trajectories for the agents are defined as $t\to \bar g_l(t)\coloneqq g_0(t)\bar g_{l0}$ with their body velocity and acceleration respectively being $\bar\xi_l=\Ad_{\bar g_{l0}^{-1}}\xi_0$ and $\dot{\bar\xi}_l=\Ad_{\bar g_{l0}^{-1}}\dot\xi_0$ ($l=1,\ldots,N_v$). Then given the configuration and velocity error for the $l$th agent $\tilde g_l$ and $\tilde\xi_l$, respectively, based on \eqref{eq:conferr}-\eqref{eq:velerr}, the tracking error dynamics of the agent can be obtained by \eqref{error-dynamicsm1}-\eqref{error-dynamicsm2}.
The locally quadratic configuration error function on $\SE$ for agent $l\in\{0,\ldots,N_v\}$ is defined by~\cite{BULLO199917, bullo-book}:
\begin{align}
\psi_l\colon\SE&\rightarrow\R_{\ge 0}\nonumber\\
\tilde g_l = \begin{bmatrix}
      \tilde R_l&\tilde p_l\\ \0_{1\times 3}&1
  \end{bmatrix}&\mapsto \frac{1}{2} \tr(\chi_l(\1_{3\times 3}-\tilde R_l))+\frac{1}{2} \tilde p_l^T K_l \tilde p_l  \label{phi_mlt},
  \end{align}
where $\tilde R_l\in\SO$ and $\tilde p_l\in\R^3$ are the rotation and position components of the $l$th agent configuration error, $\chi_l=\frac{1}{2}\tr(\Gamma_l)\1_{3\times 3}-\Gamma_l \in \R^{3 \times 3}$ is a symmetric matrix and $\Gamma_l,K_l \in \R^{3 \times 3}$ are positive-definite symmetric matrices. The smallest critical value of $\psi_l$ greater than $0$ is $\Theta_l=\lambda_{min}(\Gamma_l)+\tr(\Gamma_l)$. Note that we included the virtual leader in this definition since we use \eqref{ustar1} to control its pose. The total derivative and Hessian of this function are:
\begin{align}
    \der^L_{\tilde g_l}\psi_l &=[(\sk(\chi_l \tilde R_l)^\vee)^T~\tilde p_l^T K_l \tilde R_l],\\
    \Hes^L_{\tilde g_l}\psi_l &= \begin{bmatrix}
        \frac{1}{2}\big(\tr(\chi_l\tilde R_l)\1_{3\times 3}-\chi_l\tilde R_l\big) && \0_{3\times 3}\\
      \tilde R_l^T(K_l\tilde p_l)^\wedge \tilde R_l  && \tilde R_l^T K_l \tilde R_l
    \end{bmatrix}.
\end{align}

Given a set of control parameters for the virtual leader and vehicles, i.e., $k_{pl}$, $\Gamma_l$ and $K_l$, and based on \eqref{ideal-error-dynamicsm1}-\eqref{ideal-error-dynamicsm2}, the nominal dynamics $(\tilde g_l^*,\tilde\xi_l^*)$ for $l\in\{0,\ldots,N_v\}$ are obtained. These equations can be integrated as part of evaluating the control actions. We now equip each vehicle with a local NN-based controller \eqref{opt-con2} and the learning rules \eqref{v3} \& \eqref{w3}-\eqref{eq:gradwg1} with the input states $x_l= [\breve{\tilde{g}}_l^T~\tilde\xi^T_l~\breve{\bar{g}}^T_l~\bar\xi^T_l~\dot{\bar{\xi}}^T_l]^T$.

\subsection{Numerical Analysis}

As a numerical example, we consider a scenario where a team of three agents must reach and keep the following formation
\begin{align*}\small
    \bar g_{10}\!\!=\!\!\begin{bmatrix}
        \1_{3\times 3}\!\!\!\!\!\!\!\!\!\! & \begin{bmatrix}
            0\\0\\10
        \end{bmatrix}\\ \0_{1\times 3}\!\!\!\!\!\!\!\!\!\! & 1
    \end{bmatrix}\!\!,  \bar g_{20}\!\!=\!\!\begin{bmatrix}
        \1_{3\times 3}\!\!\!\!\!\!\!\!\!\! & \begin{bmatrix}
            -10\\0\\0
        \end{bmatrix}\\ \0_{1\times 3}\!\!\!\!\!\!\!\!\!\! & 1
    \end{bmatrix}\!\!,  \bar g_{30}\!\!=\!\!\begin{bmatrix}
        \1_{3\times 3}\!\!\!\!\!\!\!\!\!\! & \begin{bmatrix}
            0\\-10\\0
        \end{bmatrix}\\ \0_{1\times 3}\!\!\!\!\!\!\!\!\!\! & 1
    \end{bmatrix}\!\!,
\end{align*}
relative to a virtual leader moving according to \eqref{leaderkinematics}-\eqref{leaderdynamics} with $\bI_0=\1_{6\times 6}$ and the initial conditions $g_0(0)=\1_{4\times 4}$ and $\xi_0(0)=\0_{6\times 1}$. Its desired trajectory is defined as $$t\mapsto\exp([0~1~0~0~0~0]^{T\wedge})\exp(0.2\sin(\frac{t}{10})[1~1~1~2~2~2]^{T\wedge}).$$  
The nominal error dynamics of the vehicles including the virtual leader is governed by the same PD-type controller based on feedback linearization with $k_{pl}=1$, $\Gamma_l=10\1_{3\times 3}$, $K_l=5\1_{3\times 3}$, and $A_l=4\1_{6\times 6}$ for all $l\in\{0,\ldots,3\}$. The mass [kg] and moments of inertia [kg.m$^2$] of the vehicles that are unknown to the controller are set to be
\begin{align*}
    \bI_l&=\begin{bmatrix}
    \diag\{0.7+I_{1l},0.8+I_{2l},0.6+I_{3l}\} & \0_{3\times 3}\\ \0_{3\times 3} & (0.8+\mathfrak{m}_l)\1_{3\times 3}
\end{bmatrix}\\
&~~~~~~~~~~~~~~~~~~~~~~~~~~~~~~~~~~~~~~~~~~~~~+0.5\sin(\frac{2\pi t}{30})\1_{6\times 6},
\end{align*}
 where $\mathfrak{m}_l,I_{1l},I_{2l},I_{3l}\in (0,1)$ $(l=1,2,3)$ are independent random numbers. The initial conditions of the vehicles are specified as follows:
 \begin{align*}\small
      g_1(0)=\begin{bmatrix}
     \1_{3\times 3}&\begin{bmatrix}
         0\\1\\4
     \end{bmatrix}\\ \0_{1\times 3} & 1
 \end{bmatrix}, \xi_1(0)=[0~0~1~0~0~0]^T, 
 \end{align*}
 \begin{align*}\small
 g_2(0)=\begin{bmatrix}
     \exp(\begin{bmatrix}
         0.2\\-0.2\\0.2
     \end{bmatrix}^\wedge)&\begin{bmatrix}
         -5\\-1\\0
     \end{bmatrix}\\ \0_{1\times 3} & 1
 \end{bmatrix},\xi_2(0)=[0~0~0~0~0~0]^T,
 \end{align*}
 \begin{align*}\small
 g_3(0)=\begin{bmatrix}
     \exp(\begin{bmatrix}
         0.4\\-0.4\\0.4
     \end{bmatrix}^\wedge)&\begin{bmatrix}
         -1\\-5\\0
     \end{bmatrix}\\ \0_{1\times 3} & 1
 \end{bmatrix}, \xi_3(0)=[0~1~0~0~0~1]^T.
 \end{align*}

We assume that each vehicle experiences a gravitational force with the fixed acceleration of $[0~0~-10]^T$[m/s$^2$] in the inertial frame and torque due to a shift of $[0.05~0~0]^T$[m] in the center of mass. As part of $\mu_l$, a viscous force of $0.7\xi_l$ is also applied to the vehicle $l\in\{1,2,3\}$. These forces and torques are unknown to the vehicles' controllers. Additionally, an unknown disturbance $d_l$ is exerted on the vehicles whose time evolution is depicted in Figure \ref{fig:dist}. All of the vehicles lose their actuator efficiency after time $t=4$[s] in two directions, such that
for $t>4$[s]
\begin{align*}
    \bE_1&=\diag\{0.1+0.4\sin^2(\frac{2\pi t}{30}),1,1,1,0.1+0.4\sin^2(\frac{2\pi t}{30}),1\},\\
    \bE_{2,3}&=\diag\{0.5,1,1,1,0.5,1\}.
\end{align*}

We now equip each vehicle with an NN controller with 42 input neurons (24 corresponding to non-fixed elements of $\tilde g_l,\bar g_l\in\R^{4\times 4}$ and 18 for $\tilde\xi_l,\bar\xi_l,\dot{\bar{\xi}}_l\in\R^6$), $50$ hidden layer neurons, and 6 output neurons for $u_l$. We choose the following control parameters for $l\in\{1,2,3\}$:
\begin{align*}
    \varsigma_l=0.1,\rho_{1l}=800,\rho_{2l}=2,\gamma_{1l}=2,\gamma_{2l}=0.7,\alpha_{1l}=\alpha_{2l}=0.5.
\end{align*}
As a measure of matching rotation between the virtual leader and vehicles, Figure \ref{fig:angles}  depicts the rotation angles $\theta_l= \arccos(\frac{trace(R_l)-1}{2})$, where $R_l$ is the rotation component of $g_l$ for $l\in\{0,\ldots,3\}$. We also simulate the response of a PD controller with the same nominal gains implemented on Vehicle 1 which becomes unstable. Figure \ref{fig:pos} collects the position trajectories of the virtual leader and vehicles rotated in the virtual leader's coordinate frame. It is clear that despite the unknown dynamics and significant nonlinear time-dependent disturbances applied to the multi-agent system, the vehicles can make and keep the formation after approximately $2$[s], which is the settling time of the nominal error dynamics. As an example, for Vehicle 1 undergoing a more severe actuator fault, we include the trajectory of the velocity error versus the nominal velocity error in Figure \ref{fig:vel}, which exhibits a good match. At time $t=15$[s] when the actuator efficiency is dropped to its minimum value of only $10\%$, we can see the controller has difficulty keeping the error at zero. Same scenario happens at time $t=22.5$[s] when the mass and moments of inertia of the Vehicle are at their minimum value while the disturbance is at its maximum. Also, Vehicle 1's control input is depicted in Figure \ref{fig:cont} which demonstrates the adaptation of the controller to the unknown disturbances applied to the system.  
It can be seen that after learning NN, the tracking errors are bounded and the agents can keep the desired formation with bounded inputs. It should be noted that in  NN controllers the weights cannot be trained offline. However, if we choose the initial weights close to the ideal values the training period will be smoother.

   \begin{figure}[!t]
\centering
\includegraphics[width=3in]{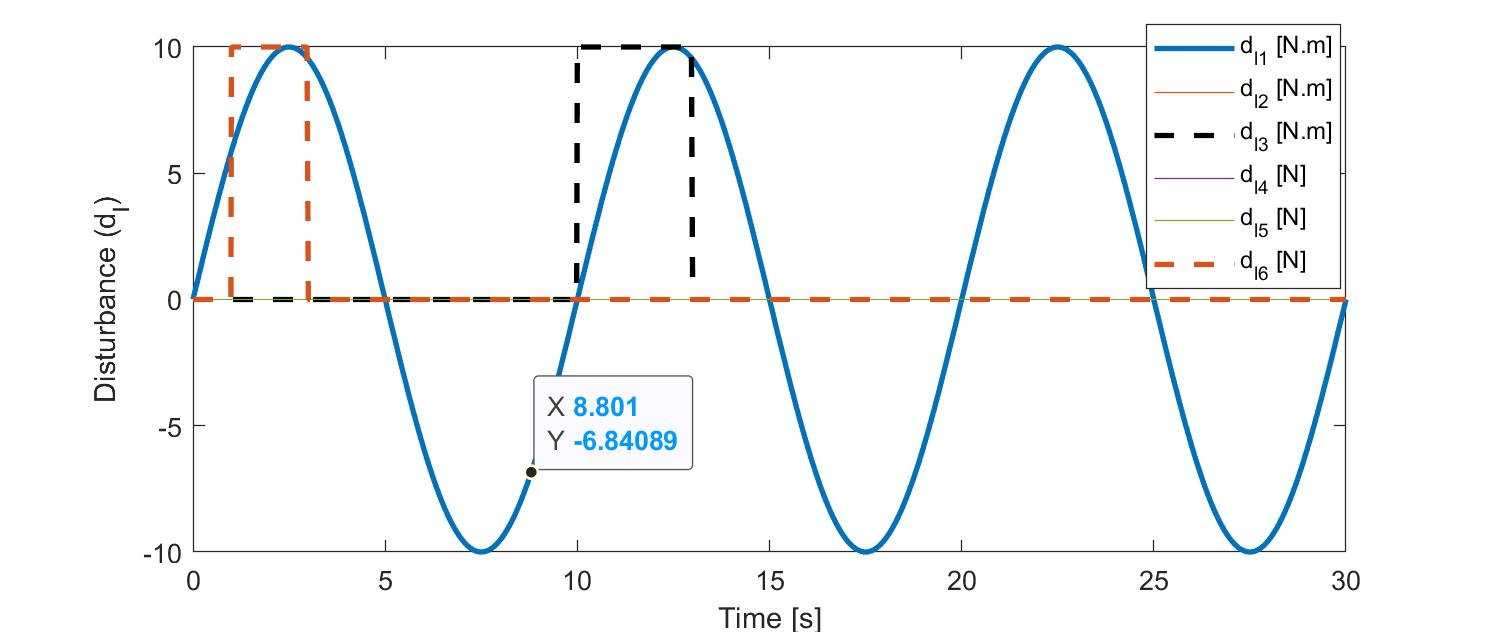}
\caption{Unknown time-dependent external disturbances}
\label{fig:dist}
\end{figure}
  \begin{figure*}[!t]
\centering
\subfloat[]{\includegraphics[width=3.5in]{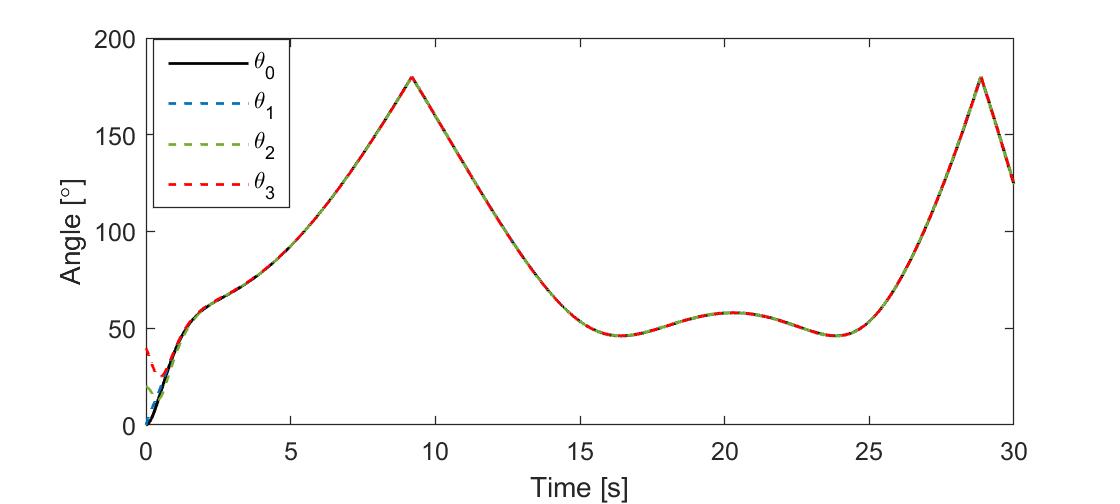}%
\label{fig:angles}}
\hfil
\subfloat[]{\includegraphics[width=3.5in]{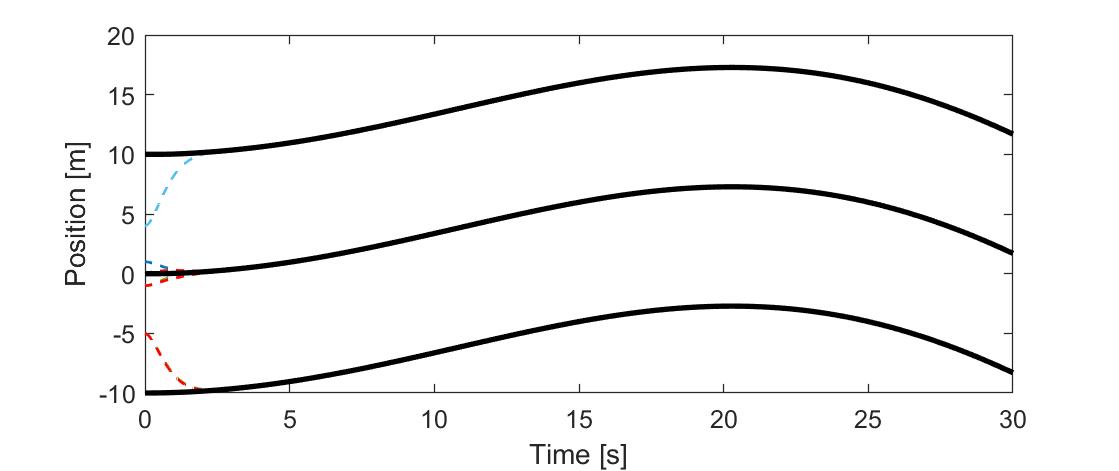}%
\label{fig:pos}}
\caption{(a) Rotation angels of vehicles;  (b) 3D position trajectories of vehicles. Curves corresponding to the virtual leader and desired offsets are black, Vehicle 1 is blue, Vehicle 2 is green, and Vehicle 3 is red.}
\end{figure*}
\begin{figure*}[thpb]
\centering
\subfloat[]{\includegraphics[width=3.5in]{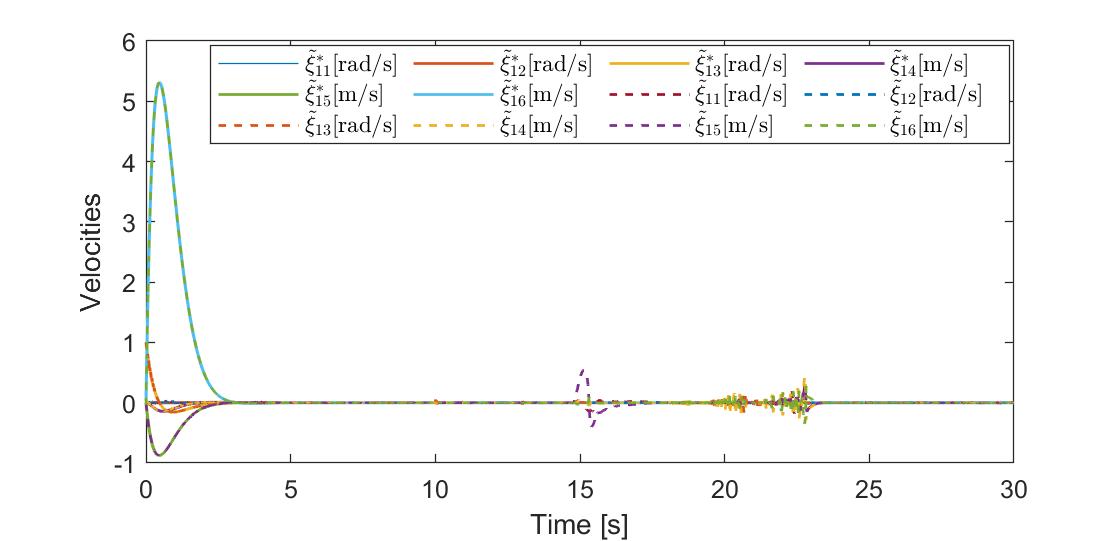}%
\label{fig:vel}}
\hfil
\subfloat[]{\includegraphics[width=3.5in]{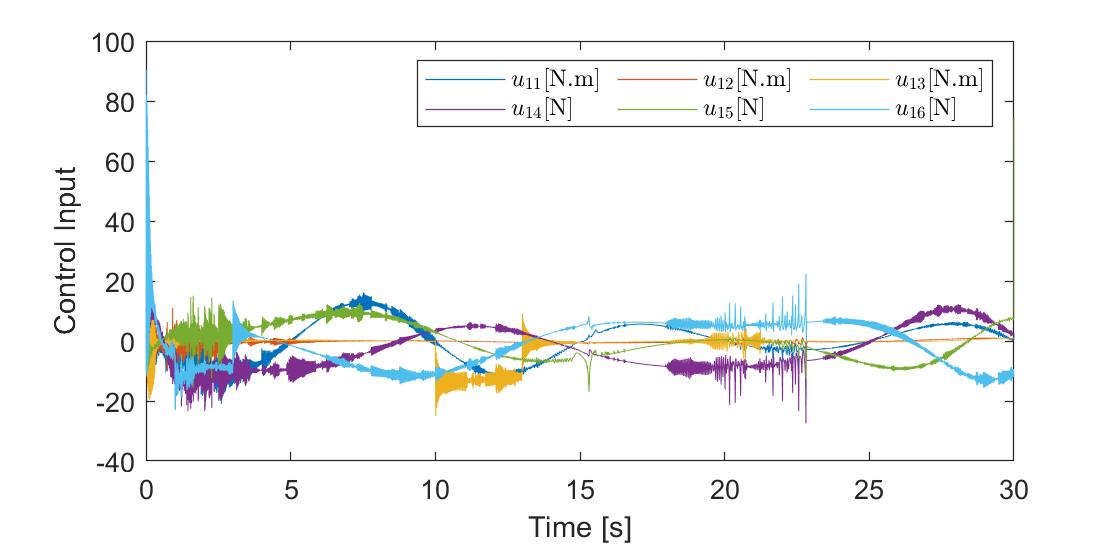}
\label{fig:cont}}
\caption{(a) Velocity error for Vehicle 1 when compared with the nominal error dynamics;  (b) Control input for vehicle 1.}
\end{figure*}
%
%
 %
\section{Conclusions} \label{conclu}
In this paper, a geometric NN tracking control was proposed for a system with unknown and disturbed dynamics on matrix Lie groups. The intrinsic learning rules were developed to minimize tracking error functions defined according to the Lie group structure of the configuration manifold. By formulating the learning rules on the Lie algebra and eliminating the need for parameterization, the approach reduces computational complexity while preserving the geometric structure of the system. We proved the uniform ultimate boundedness of the tracking error and NN weights, ensuring almost global stability under mild assumptions.

The proposed learning and control framework applies to a wide range of systems, including mobile robots. By leveraging the left invariance of the tangent bundle, our method ensures that velocities remain independent of the configuration, enhancing robustness and compatibility with onboard sensor measurements. Furthermore, as control parameters are updated online, the approach can adapt to changing operational conditions, unmodeled dynamics, bounded disturbances, and even limited structural damage while maintaining stability and acceptable performance. Future work will focus on extending this NN controller to nonholonomic underactuated systems evolving on matrix Lie groups, with experimental validation.

\bibliographystyle{ieeetr}
\bibliography{sample}
\appendices
\section{Proof of Proposition 1}\label{app:pro-prop1}
The proof consists of three steps. \textbf{Step 1:} Based on Lemma \ref{lem:smooth}, there exists at least an extension $\bar f$ for the function $f$. \textbf{Step 2:}  Based on the definition in \eqref{eq:dir-der} and the fact that the directional derivative is independent of the curve: 
    \begin{align*}
        \left<\der_g^L f,\eta_i\right>=\left<\der_gf|g\eta_i^\wedge\right>=\left.\frac{d}{d\epsilon}\right|_{\epsilon=0}\!\!\!\!\!\!\!\!\!\! f(g\exp(\epsilon \eta_i^\wedge)).
    \end{align*}
We chose the curve $g(\epsilon)=g\exp(\epsilon \eta_i^\wedge)$ to evaluate the directional derivative. \textbf{Step 3:} Since $G$ is an embedded manifold in $\R^{N\times N}$, the inclusion map $\iota\colon G\hookrightarrow \R^{N\times N}$ is smooth and its push forward $\iota_*$ at every $g\in G$ is the restriction of the tangent space $T_g\R^{N\times N}$ to $T_gG$. Therefore, 
    \begin{align*}
        \left<\der_gf|g\eta_i^\wedge\right>&=\left<\der_g(\bar f\circ \iota)|g\eta_i^\wedge\right> = \left<\iota^*\der_g\bar f|g\eta_i^\wedge\right>\\ 
        &= \left<\der_g\bar f| \iota_*g\eta_i^\wedge\right> = \left<\der_g\bar f| g\eta_i^\wedge\right>=\left.\frac{d}{d\epsilon}\right|_{\epsilon=0}\!\!\!\!\!\!\!\!\!\! \bar f(g+\epsilon g\eta_i^\wedge).
    \end{align*}
    Here, we use the fact that $\bar f$ is the (local) extension of $f$, i.e., $f=\bar f\circ \iota$, in the second equality. Fourth equality is true, since $g\eta_i^\wedge\in T_gG$. Finally, the last equality is the computation of the directional derivative in the ambient manifold $\R^{N\times N}$. 

    \section{Proof of Lemma 3 \& 4}\label{app:lem3&4}
    We first prove Lemma \ref{lem:gradW}. We take the left translated gradient of $\tilde g$ with respect to an arbitrary element $W_{ij}$ using Definition \ref{def:JAcobian}. Note that the natural basis for $\fg$ induced by $\wedge$ is assumed. 
    \begin{align*}
        \frac{\partial}{\partial t} \nabla^L_{{\hat{W}}_{ij}}\tilde g &= \frac{\partial}{\partial t} \big(\tilde g^{-1}\frac{\partial \tilde g}{\partial \hat{W}_{ij}}\big)^\vee = \big(-\tilde g^{-1}\dot{\tilde g}\tilde g^{-1} \frac{\partial \tilde g}{\partial \hat{W}_{ij}} +  \tilde g^{-1} \frac{\partial \dot{\tilde {g}}}{\partial \hat{W}_{ij}} \big)^\vee \\
        &=\big(-\tilde\xi^\wedge (\nabla^L_{{\hat{W}}_{ij}}\tilde g)^\wedge +  (\nabla^L_{{\hat{W}}_{ij}}\tilde g)^\wedge\tilde\xi^\wedge + (\nabla_{{\hat{W}}_{ij}}\tilde \xi)^\wedge \big)^\vee\\
        &=-\ad_{\tilde\xi} (\nabla^L_{{\hat{W}_{ij}}}\tilde g)+\nabla_{{\hat{W}_{ij}}}\tilde\xi\in\R^{n}.
        \end{align*}
        As for $\tilde\xi$, we use \eqref{er-2} and by Definition \ref{def:der} and \eqref{eq:Der}:
        \begin{align*}
       \frac{\partial}{\partial t}& \nabla_{{\hat{W}_{ij}}}\tilde \xi = \frac{\partial \dot{\tilde \xi} }{\partial \hat{W}_{ij}}= -A \frac{\partial {\tilde \xi} }{\partial \hat{W}_{ij}}-k_p\Der^L_{\tilde g}(\der^L_{\tilde g}\psi)\big(\tilde g^{-1} \frac{\partial {\tilde g} }{\partial \hat{W}_{ij}}\big)^\vee\\
       &+\bB\stackrel{\mbox{$--i--~~$}}{[0\cdots \sigma_j\cdots 0]^T}\\
       &=-A \nabla_{\hat{W}_{ij}}\tilde \xi - k_p(\Hes^L_{\tilde g}\psi)\nabla^L_{\hat{W}_{ij}}\tilde g+\bB\stackrel{\mbox{$--i--~~$}}{[0\cdots \sigma_j\cdots 0]^T}\in\R^{n}.
    \end{align*}
The full matrix equations can be formed by assembling all such columns for the vector $\breve{\hat{W}}\in\R^{nm}$ according to \eqref{eq:Jacobian} and the definition of \textit{breve} operator in \eqref{eq:breve}.

Proof of Lemma 4 follows the same lines for an arbitrary $\hat{V}_{ij}$, except that \eqref{eq:gradxi} holds since
\begin{align*}
    \frac{\partial \big(\hat{W}\sigma(\hat{V}x)\big)}{\partial \hat{V}_{ij}}= \left.\hat{W}\nabla_{y}\sigma\right|_{y=\hat{V}x}\stackrel{\mbox{$--i--~~$}}{[0\cdots x_j\cdots 0]^T}\in\R^{n},
\end{align*}
and based on the definition of $\sigma$ in \eqref{eq:sigma} we have
$
    \frac{\partial\sigma_i}{\partial y_i}=\frac{4e^{-2y_i}}{(1+e^{-2y_i})^2}=1-\sigma_i^2.
$

\section{Proof of Proposition 2}\label{app:prop2}
To prove this proposition we have to find $\nabla_{\hat{W}} F$. We compute it by focusing on each term of $F$, separately. 


For the first term, we use \eqref{eq:gradwxi} and Assumptions \ref{as:staticapprox} and \ref{as:ApproxInertia} to first compute 
\begin{align}\label{eq:gradxistat}
    \nabla_{\breve{\hat {W}}} \tilde \xi= -k_pA^{-1}(\Hes^L_{\tilde g}\psi)\nabla^L_{\breve{\hat{W}}}\tilde g
       +A^{-1}\left[
           \frac{\sigma_1}{\varsigma}\1_{n\times n} ~\cdots~\frac{\sigma_m}{\varsigma}\1_{n\times n}
       \right]_{y=\hat Vx}.
\end{align}
Then, since $A$ is symmetric and positive-definite:
\begin{align*}
    \nabla_{\hat W}&\big(\frac{1}{2}(\tilde\xi-\tilde\xi^*)^TA(\tilde\xi-\tilde\xi^*)\big) = \Big(\big(\nabla_{\breve{\hat {W}}} \tilde \xi\big)^TA(\tilde\xi-\tilde\xi^*)\Big)\breve\breve = \\
    &-k_p\left[\Big((\Hes^L_{\tilde g}\psi)\nabla^L_{\breve{\hat{W}}}\tilde g\Big)^T(\tilde\xi-\tilde\xi^*)\right]\breve\breve \\ &+\left[
           \frac{\sigma_1}{\varsigma}(\tilde\xi-\tilde\xi^*)^T ~\cdots~\frac{\sigma_m}{\varsigma}(\tilde\xi-\tilde\xi^*)^T\right]^T\breve\breve \\ &= -k_p\left[\Big((\Hes^L_{\tilde g}\psi)\nabla^L_{\breve{\hat{W}}}\tilde g\Big)^T(\tilde\xi-\tilde\xi^*)\right]\breve\breve + \frac{1}{\varsigma}(\tilde\xi-\tilde\xi^*)\sigma^T.
\end{align*}

The streightforward computation of the second term is based on \eqref{eq:defgradF} and Definition \ref{def:der}:
$$\nabla_{\hat{W}}\psi=\left[\big(\der_{(\tilde g^*)^{-1}\tilde g}^L\psi\big)\nabla^L_{\breve{\hat{W}}}\tilde g \right]^T\breve\breve. $$
Note that according to Definition \ref{def:JAcobian} and the fact that the nominal error dynamics is independent of $\hat W$, we have $\nabla^L_{\breve{\hat{W}}}(\tilde g^*)^{-1}\tilde g=\nabla^L_{\breve{\hat{W}}}\tilde g$.  With some manipulation, \eqref{w3} is obtained, which is dependent on the propagation of $\nabla^L_{\breve{\hat{W}}}\tilde g$. This term can be solved in time using \eqref{eq:gradwg}. It is easy to see that under Assumptions \ref{as:staticapprox} and \ref{as:ApproxInertia}, we can substitute \eqref{eq:gradxistat} into \eqref{eq:gradwg} to obtain \eqref{eq:gradwg1}.
\section{Proof of Lemma 5}\label{app:lem5}
Consider the control signal \eqref{opt-con2} and the element $\hat V_{ij}$:
    \begin{align*}
       \Big(\nabla_{{\hat{V}}}\Big(\frac{1}{2}u^T u\Big)\Big)_{ij}&= \frac{\partial}{\partial \hat{V}_{ij}}\Big(\frac{1}{2}u^T u\Big)=u^T\frac{\partial u}{\partial \hat{V}_{ij}} \\
       &=\overbrace{\sigma^T\hat{W}^T\hat{W}\Pi}^{D}\stackrel{\mbox{$--i--~~$}}{[0\cdots x_j\cdots 0]^T}=D_ix_j.\\
    \end{align*}
    Therefore, in matrix form we have
      $  \nabla_{{\hat{V}}}\Big(\frac{1}{2}u^T u\Big) = D^Tx^T$, which completes the proof knowing that $\Pi^T=\Pi$.

    \section{The Special Euclidean Group $\SE$}\label{app:SE(3)}
        The Special Euclidean group $\SE$ is an example of a $6$-dimensional matrix Lie group that is embedded in both  $\mathbb{GL}(4)$ and $\R^{4\times 4}$, whose elements are in the form of 
\begin{align*}
g=\begin{bmatrix}
 R & p \\
\mathbf{0}_{1\times 3} & 1 \\
\end{bmatrix}
 \in \SE\subset\R^{4\times 4}\cong\R^{16}. \label{g1}
 \end{align*}
Here, $R$ is a rotation matrix in the Special Orthogonal group $\SO:=\{R\in\R^{3\times 3}|\,\,R^TR=\1_{3\times 3},\,\, \det(R)=1\}$ and $p\in\R^3$ is a position vector. Note that since the last row of the elements of $\SE$ is fixed, one can embed $\SE$ into the smaller vector space of $\R^{12}$.
     The Lie algebra of the Special Euclidean group is denoted by $\se$, which is a 6-dimensional vector space. Using the isomorphism $\wedge\colon\R^6\rightarrow\se$, the elements of $\se$ are identified as 
$$\xi^\wedge\coloneqq 
 \begin{bmatrix}
{\omega}^\wedge & v \\
\mathbf{0}_{1\times 3} & 0 \\
 \end{bmatrix}
 \in {\se} \subset \mathbb{R}^{4 \times 4},$$
for every $\xi=[\omega^T~~ v^T]^T\in\R^6$ that includes the angular ($\omega\in\R^3$) and linear ($v\in\R^3$) velocity components. For every $\omega, r\in\R^3$ the vector space isomorphism $\wedge\colon\R^3\rightarrow\so\subset\R^{3 \times 3}$ is defined by $\omega^\wedge r=\omega\times r$.
Accordingly using the canonical basis of $\R^3$, we have a natural basis for $\so$ and hence for $\se$.

For $\SE$, the exponential map is computed through
\begin{align*}
    e^{{\xi}^\wedge}&=\begin{bmatrix}\1_{3\times 3} & v \\ \0_{1\times 3} & 1\end{bmatrix}, \,\,\, \omega=\0_{3\times 1}\\
    e^{\xi^\wedge}&=\begin{bmatrix} e^{\omega^\wedge}& (\1_{3\times 3}-e^{\omega^\wedge})\frac{\omega^\wedge v}{\|\omega\|^2}+\frac{\omega {\omega}^T v}{\|\omega\|^2} \\ \0_{1\times 3} &1\end{bmatrix}, \,\,\, \omega\neq \0_{3\times 1}
\end{align*}
where 
\begin{align*}
     e^{\omega^\wedge}=\1_{3\times 3}+\frac{\omega^\wedge}{\|\omega\|}\sin(\|\omega\|)+ \frac{(\omega^\wedge)^2}{\|\omega\|^2}\big(1-\cos (\|\omega\|)\big).
\end{align*}
Further, the adjoint representations of the group and its Lie algebra are
    $$\Ad_g =\begin{bmatrix}
        R && \0_{3\times 3} \\ p^\wedge R && R
    \end{bmatrix}, ~~~~~ \ad_\xi=\begin{bmatrix}
        \omega^\wedge && \0_{3\times 3} \\ v^\wedge && \omega^\wedge
    \end{bmatrix}.$$
\begin{IEEEbiography}[{\includegraphics[width=1in,height=1.25in,clip,keepaspectratio]{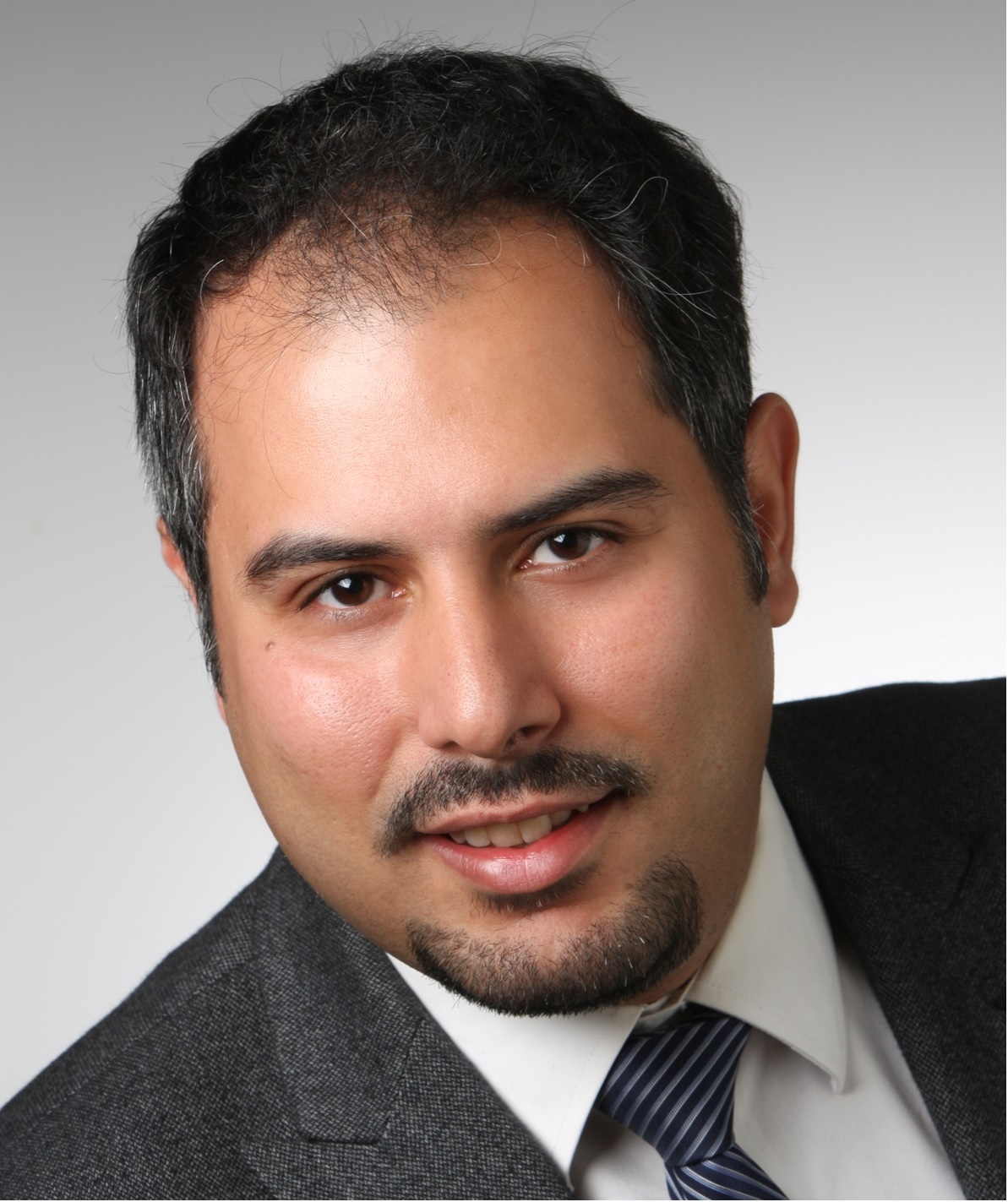}}]{Robin Chhabra} (S'25-M'16) is an assistant professor with the Department of Mechanical, Industrial, and Mechatronics Engineering, Toronto Metropolitan University. He also served as a
Tier-II Canada Research Chair in autonomous space robotics and mechatronics at Carleton University. Dr. Chhabra received his M.A.Sc. (2008) and Ph.D. (2013) from the University of Toronto. After a postdoctoral fellowship, he joined MacDonald Dettwiler
and Associates Ltd., where he
performed applied research on the control of space systems. Dr.
Chhabra’s team researches geometric mechanics, nonlinear controls, and artificial intelligence to develop novel methodologies for the general intelligence of robotic systems.
\end{IEEEbiography}
\begin{IEEEbiography}
[{\includegraphics[width=1in,height=1.25in,clip,keepaspectratio]{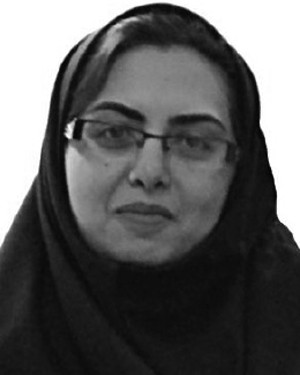}}]
{Farzaneh Abdollahi} (S'98–M'09) received the B.Sc. degree in electrical engineering from the Isfahan University of Technology, Isfahan, Iran, in 1999, the M.Sc. degree in electrical engineering from the Amirkabir University of Technology, Tehran, Iran, in 2003, and the Ph.D. degree in electrical engineering from Concordia University, Montreal, QC, Canada, in 2008. She is currently an Associate Professor with the Amirkabir University of Technology and a Research Assistant Professor with Carleton University. Her research interests include neural networks, robotics, control of nonlinear systems, control of multiagent networks, and robust and switching control.

\end{IEEEbiography}

%
%
%
%


\vfill

\end{document}